\documentclass[12pt,a4paper]{elsarticle}
\usepackage[nottoc]{tocbibind}
\usepackage{graphicx}
\usepackage{slashed}
\usepackage{enumerate}
\usepackage{amsmath}
\usepackage{mathtools,cancel}
\usepackage[bottom]{footmisc}
\usepackage[scr=rsfso]{mathalfa}
\usepackage[hidelinks]{hyperref}
{

}
\usepackage{amsfonts}
\usepackage{psfrag}
\usepackage{color}
\usepackage{mathtools}
\usepackage{pgf,tikz}
\usepackage{mathrsfs}
\usetikzlibrary{arrows}
\usepackage{fancyhdr}
\usepackage{amssymb}
\usepackage{slashed}
\usepackage{enumitem}
\usepackage{pst-node}
\usepackage{tikz-cd} 
\usepackage[T1]{fontenc}
\usepackage[utf8]{inputenc}

\newtheorem{definition}{Definition}
\newtheorem{remark}{Remark}

\newtheorem{proposition}{Proposition}[section]

\newtheorem{theorem}{Theorem}[section]
\newtheorem{corollary}{Corollary}[section]
\newtheorem{lemma}{Lemma}[section]

\numberwithin{equation}{section}
\allowdisplaybreaks[4]

\newenvironment{proof}{\smallskip\noindent\emph{Proof.}\hspace{1pt}}%
{\hspace{-5pt}{\nobreak\quad\nobreak\hfill\nobreak$\square$\vspace{8pt}%
		\par}\smallskip\goodbreak}

\newcommand{\hnabla}{\widehat{\nabla}}

\newcommand{\be}{\begin{equation}}
\newcommand{\ee}{\end{equation}}

\newcommand{\bm}{\begin{align*}}
\newcommand{\enm}{\end{align*}}

\newcommand{\bespeq}{\begin{equation}\begin{split}}
\newcommand{\espeq}{\end{split}\end{equation}}

\newcommand{\tr}{\mbox{tr}}

\newcommand\restri[2]{{
		\left.\kern-\nulldelimiterspace 
		#1 
		\right|_{#2} 
}}

\definecolor{ffqqqq}{rgb}{1.,0.,0.}
\definecolor{uuuuuu}{rgb}{0.26666666666666666,0.26666666666666666,0.26666666666666666}

\makeatletter
\def\ps@pprintTitle{%
  \let\@oddhead\@empty
  \let\@evenhead\@empty
  \let\@oddfoot\@empty
  \let\@evenfoot\@oddfoot
}
\def\@author#1{\g@addto@macro\elsauthors{\normalsize%
    \def\baselinestretch{1}%
    \upshape\authorsep#1\unskip\textsuperscript{%
      \ifx\@fnmark\@empty\else\unskip\sep\@fnmark\let\sep=,\fi
      \ifx\@corref\@empty\else\unskip\sep\@corref\let\sep=,\fi
      }%
    \def\authorsep{\unskip,\space}%
    \global\let\@fnmark\@empty
    \global\let\@corref\@empty  
    \global\let\sep\@empty}%
    \@eadauthor={#1}
}
\makeatother

\begin{document}

\thispagestyle{empty}

\begin{frontmatter}
	
\title{Quasi-local masses in General relativity and their positivity: Spinor approach}
\author{Puskar Mondal\footnote{puskar\_mondal@fas.harvard.edu}, Shing-Tung Yau \footnote{yau@math.harvard.edu} 
\\ Department of Mathematics, Harvard University,\\
Center of Mathematical Sciences and Applications, Harvard University}

\begin{abstract}
\noindent We study the quasi-local masses arising in general relativity using spinors and prove their positivity property. This leads to the question of a pure quasi-local proof of the positivity of the Wang-Yau \cite{yau} quasi-local mass. More precisely we prove that the gravitational mass bounded by a spacelike topological $2-$sphere is non-negative in a generic spacetime verifying dominant energy condition and vanishes only if the surface is embedded in the Minkowski space. This construction is purely quasi-local in nature and in particular does not rely on Bartanik's gluing and asymptotic extension construction \cite{bartnik1993quasi} and subsequent application of the positive mass theorem \cite{schoen1979proof,schoen1981proof} to prove the positivity of quasi-local mass. The result involves solving Dirac equation on a compact Riemannian manifold with boudary using MIT Bag and APS boundary condition.

\end{abstract}
\end{frontmatter}

\setcounter{tocdepth}{2}
{\hypersetup{linkcolor=black}
\small
\tableofcontents
}

\section{Introduction}
\noindent Wang-Yau quasi-local mass is defined for a Lorentzian manifold. We are given a four-dimensional $C^{\infty}$ globally hyperbolic Lorentzian manifold $(M,\widehat{g})$ with signature $-,+,+,+$ and a symmetric rank-2 tensor $T_{\mu\nu}$ (stress-energy tensor of the sources present in the spacetime) that is related to the Ricci and scalar curvature of $(M,\widehat{g})$ via the Einstein equations 
\begin{eqnarray}
R_{\mu\nu}-\frac{1}{2}R\widehat{g}_{\mu\nu}=T_{\mu\nu},   
\end{eqnarray}
where $8\pi G$ is set to 1 ($G$ being Newton's constant). The stress-energy tensor is assumed to verify the dominant energy condition i.e., 
\begin{eqnarray}
    T(e_{0},e_{0})\geq |T(e_{0},e_{i})|
\end{eqnarray}
for a local orthonormal frame $(e_{0},e_{i})_{i=1}^{3}$. One can make a better sense of the dominant energy condition if we adopt a $3+1$ formulation. Consider a spacelike topological ball $\Omega$ and $n$, a time-like unit normal field to it. In local coordinates $(t,x^{i})_{i=1}^{3}$, $n=\frac{1}{N}(\partial_{t}-X)$ where $X$ is the shift vector field parallel to $\Omega$. Spacetime metric $\widehat{g}$ can be decomposed as 
\begin{eqnarray}
 \widehat{g}:=-N^{2}dt\otimes dt+g_{ij}(dx^{i}+X^{i}dt)\otimes (dx^{j}+X^{j}dt),   
\end{eqnarray}
 where $g_{ij}:=\widehat{g}(\partial_{i},\partial_{j})$ is the induced metric on $\Omega$ by the Lorentzian structure of $(M,\widehat{g})$. Let $K_{ij}:=\widehat{g}(\nabla[\widehat{g}]_{\partial_{i}}n,\partial_{j})$ be the second fundamental form of $\Omega$ in $M$. We write the Einstein constraint equations which will be of importance to us 
 \begin{eqnarray}
    R(g)-K_{ij}K^{ij}+(\text{tr}_{g}K)^{2}=2\mu\\
 \widehat{\nabla}^{j}K_{ij}-\nabla_{i}\text{tr}_{g}K=J_{i}, 
 \end{eqnarray}
 where $\mu:=T(n,n)$ is the energy density and $J_{i}=T(n,\partial_{i})$ is the momentum density. $R(g)$ is the scalar curvature of $(\widehat{\Omega},g)$. The dominant energy condition implies 
 \begin{eqnarray}
 \label{eq:dominant}
   \mu\geq \sqrt{g(J,J)}.  
 \end{eqnarray}
The physical interpretation of this would be that the energy-momentum vector $(\mu,J^{i})$ is time-like or essentially the finite propagation of speed. An important special case is when $K_{ij}=0$ (time-symmetric case) and the dominant energy condition implies that the scalar curvature of $\Omega$ is non-negative. 

\noindent Let $(\Sigma, \sigma)$ be a topological $2-$sphere in the physical spacetime $(M,\widehat{g})$ with induced metric $\sigma$ bounding the spacelike ball $\Omega$. The spacetime $(M,\widehat{g})$ is assumed to verify desirable properties described in the beginning of the introduction. In this framework, one asks ``what is the amount of gravitational energy contained within the ball $\Omega$ with boundary $ \widehat{\Sigma}$"? Some desirable physical properties of the definition of energy should include non-negativity under the assumption of an appropriate energy condition on the spacetime and rigidity i.e., it should vanish for a topological $2-$sphere embedded in the Minkowski space. In addition, it should encode the information about pure gravitational energy (Weyl curvature effect) as well as the stress-energy tensor of any source fields present in the spacetime. Based on a Hamilton-Jacobi analysis, Brown-York \cite{brown1992quasilocal, brown1993quasilocal} and Liu-Yau \cite{liu2003positivity,liu2006positivity} defined a quasi-local mass by isometrically embedding the $2-$surface into the reference Euclidean $3-$ space and comparing the extrinsic geometry (the formulation relied on the embedding theorem of Pogorelov \cite{pogorelov1952regularity} i.e., the topological $2-$sphere needed to posses everywhere positive sectional curvature).  
\begin{definition}[Brown-York Mass] 
Let $\Sigma$ be a topological $2-$sphere in the spacetime $(M,\widehat{g})$ bounding a spacelike topological ball $\Omega$.
Consider isometric embedding of $(\Sigma, \sigma)$ into the Euclidean space $(\mathbb{R}^{3},\delta)$ $i:(\Sigma, \sigma)\to (i(\Sigma), \sigma)$. Let $\widehat{k}$ be the mean curvature of $\Sigma$ while viewed as an embedded surface in $B$ in the physical spacetime $(M,\widehat{g})$.
  Let $k_{0}$ be the same for its image $i( \widehat{\Sigma})$ in $\mathbb{R}^{3}$. Also assume that the Gaussian curvature of $(\Sigma, \sigma)$ is strictly positive. Then the Brown-York mass is defined as 
\begin{eqnarray}
 M^{BY}:=\frac{1}{8\pi}\int_{\Sigma}(k_{0}-\widehat{k})\mu_{\sigma}.   
\end{eqnarray}  
\end{definition}

The Liu-Yau mass is defined similarly with mean curvature $k$ replaced by the Lorentzian norm of the spacetime mean curvature of $\Sigma$ in $(M,\widehat{g})$ i.e., 
\begin{definition}[Li-Yau Mass]
Assume that the topological $2-$sphere $( \Sigma,\sigma)$ bounds a topogical ball $\Omega$ in the spacetime $(M,\widehat{g})$. Also assume that the Gaussian curvature of $( \Sigma, \sigma)$ is strictly positive. The Liu-Yau mass is defined as follows
 \begin{eqnarray}
 M^{LY}:=\frac{1}{8\pi}\int_{\Sigma}(k_{0}-|H|)\mu_{\sigma},   
\end{eqnarray}   
where $|H|$ is the Lorentzian norm of the mean curvature vector of $(\Sigma,\sigma)$ in the spacetime $(M,\widehat{g})$.
\end{definition}

Both of these two masses are positive. But they turn out to be too positive. \cite{murchadha2004comment} discovered surfaces in Minkowski space with strictly positive Brown-York and Liu-Yau mass. It appeared that these formulations lacked an accurate prescription of spacetime momentum information. This led Wang and Yau \cite{yau,yau1} to define the most consistent notion of the quasi-local mass associated with a space-like topological $2-$sphere (we should mention that the second author together with Alaee and Khuri recently provided a new definition of quasi-local mass bounded by a class of surfaces verifying certain topological conditions in \cite{alaee2023quasi}). We discuss the definition of this quasi-local mass.

\noindent Let us assume that the mean curvature vector $\mathbf{H}$ of $\Sigma$ is space-like. Let $\mathbf{J}$ be the reflection of $\mathbf{H}$ through the future outgoing light cone in the normal bundle of $ \Sigma$. The data that Wang and Yau use to define the quasi-local energy is the triple ($\sigma,|\mathbf{H}|_{\hat{g}},\alpha_{\mathbf{H}}$) on $\Sigma$, where $\sigma$ is the induced metric on $\Sigma$ by the Lorentzian metric $\hat{g}$ on the physical spacetime $M$, $|\mathbf{H}|_{\hat{g}}$ is the Lorentzian norm of $\mathbf{H}$, and $\alpha_{\mathbf{H}}$ is the connection $1-$form of the normal bundle with respect to the mean curvature vector $\mathbf{H}$ and is defined as follows 
\begin{eqnarray}
\label{eq:connection}
\alpha_{\mathbf{H}}(X):=\hat{g}(\nabla[\hat{g}]_{X}\frac{\mathbf{J}}{|\mathbf{H}|},\frac{\mathbf{H}}{|\mathbf{H}|}).
\end{eqnarray}
Choose a basis pair ($e_{3},e_{4}$) for the normal bundle of $\Sigma$ in the spacetime that satisfy $\hat{g}(e_{3},e_{3})=1,\hat{g}(e_{4},e_{4})=-1$, and $\hat{g}(e_{3},e_{4})=0$. Now embed the $2-$surface $\Sigma$ isometrically into the Minkowski space with its usual metric $\eta$ i.e., the embedding map $X: x^{a}\mapsto X^{\mu}(x^{a})$ satisfies 
\begin{eqnarray}
\sigma(\frac{\partial}{\partial x^{a}},\frac{\partial}{\partial x^{b}})=\eta(\frac{\partial X}{\partial x^{a}},\frac{\partial X}{\partial x^{b}}),
\end{eqnarray} where $\{x^{a}\}_{a=1}^{2}$ are the coordinates on $\Sigma$. Now identify a basis pair ($e_{30},e_{40}$) in the normal bundle of $X(\Sigma)$ in the Minkowski space that satisfies the exact similar property as $(e_{3},e_{4})$. In addition the time-like unit vector $e_{4}$ is chosen to be future directed i.e., $\hat{g}(e_{4},\partial_{t})<0$. Let $\tau:=-\langle X,\partial_{t}\rangle_{\eta}$, a function on $ \Sigma$ be the time function of the embedding $X$. 
The Wang-Yau quasi-local energy is defined as follows 
\begin{definition}[Wang-Yau Energy]
\label{wy14}
Let $(M,\widehat{g})$ be the physical spacetime that verifies the dominant energy condition. Suppose $\Sigma$ is an embedded topological $2-$sphere in $(M,\widehat{g})$ with spacelike mean curvature vector $\mathbf{H}$. Let $i_{0}: \Sigma \hookrightarrow \mathbb{R}^{1,3}$ be an isometric embedding of $(\Sigma,\sigma)$ into the Minkowski space and $\tau$ be the restriction of the time function $t$ to $i_{0}(\Sigma)$. Let $\mathbf{H}_{0}$ be the mean curvature vector of the embedding $i_{0}(\Sigma)$ in the Minkowski space. If $\tau$ is admissible (see definition \ref{admi}), then the Wang-Yau quasi-local energy is defined to be
 \begin{eqnarray}
\label{eq:wy1}
8\pi M^{WY}:=\inf_{\tau}\left(\underbrace{\int_{\Sigma}\left(-\sqrt{1+|\nabla\tau|^{2}_{\sigma}}\langle \mathbf{H}_{0},e_{30}\rangle-\langle\nabla[\eta]_{\nabla\tau}e_{30},e_{40}\rangle\right)\mu_{\sigma}}_{I:=Contribution~from~the~Minkowski~space}\right.\\\nonumber 
\left.-\underbrace{\int_{\Sigma}\left(-\sqrt{1+|\nabla\tau|^{2}_{\sigma}}\langle \mathbf{H},e_{3}\rangle-\langle\nabla[\hat{g}]_{\nabla\tau}e_{3},e_{4}\rangle\right)\mu_{\sigma}}_{II:=contribution~from~the~physical~spacetime}\right).
\end{eqnarray} 
Here $\alpha_{H_{0}}$ is the connection of the normal bundle of $i_{0}(\Sigma)$ in the Minkowski space.
\end{definition}
We will be interested in proving the non-negativity of the following entity we call \textit{reduced} Wang-Yau energy 
\begin{eqnarray}
\label{eq:wy2}
8\pi M^{wy}_{\tau}:=\left(\int_{\Sigma}\left(-\sqrt{1+|\nabla\tau|^{2}_{\sigma}}\langle \mathbf{H}_{0},e_{30}\rangle-\langle\nabla[\eta]_{\nabla\tau}e_{30},e_{40}\rangle\right)\mu_{ \sigma}\right.\\\nonumber 
\left.-\int_{\Sigma}\left(-\sqrt{1+|\nabla\tau|^{2}_{\sigma}}\langle \mathbf{H},e_{3}\rangle-\langle\nabla[\hat{g}]_{\nabla\tau}e_{3},e_{4}\rangle\right)\mu_{\sigma}\right).    
\end{eqnarray}
for every $\tau$ being an admissible time function (see definition \ref{admi}).
There is still a boost-redundancy left in the normal bundle of $\Sigma$ in $N$. Let us illustrate this. Till now the choice of the normal bundle basis $(e_{3},e_{4})$ is not fixed (same for the Minkowski basis $(e_{30},e_{40})$). Therefore, one can equally choose another basis pair $(e^{'}_{3},e^{'}_{4})$ that is related to the old ones by a boost i.e., 
\begin{eqnarray}
 e^{'}_{3}=\cosh\varphi e_{3}+\sinh \varphi e_{4}\\
 e^{'}_{4}=\sinh\varphi e_{3}-\cosh \varphi e_{4}
\end{eqnarray}
Therefore, the energy is a function of $\varphi$. Fortunately, it is a convex function of $\varphi$, and minimization is possible. 
\cite{yau} minimizes the physical space contribution $II$ over all possible boost-angles. This lets one get rid of the boost redundancies present in the definition. See section 2 of \cite{yau} for a detailed exposition.

\noindent This Wang-Yau quasi-local mass possesses several good properties that are desired on physical ground. This mass is strictly positive for $2-$surfaces bounding a space-like domain in a curved spacetime that satisfies the dominant energy condition and it identically vanishes for any such $2-$surface in Minkowski spacetime. In addition, it coincides with the ADM mass at space-like infinity \cite{wang2010limit} and Bondi-mass at null infinity \cite{chen2011evaluating} and reproduces the time component of the Bel-Robinson tensor (a pure gravitational entity) together with the matter stress-energy at the small sphere limit, that is when the $2-$sphere of interest is evolved by the flow of its null geodesic generators and the vertex of the associated null cone is approached \cite{chen2018evaluating}. In addition, explicit conservation laws were also discovered at the asymptotic infinity \cite{chen2015conserved}.

Wang-Yau energy can be written in a more familiar form. This reduction is actually vital in proving its positivity. This is achieved in two steps (see theorem \ref{mean1} and \ref{mean2} for the details). First, consider the isometric image of $ \Sigma$ $i_{0}(\Sigma)$ ($i_{0}: \Sigma\to \mathbb{R}^{1,3}, \Sigma\mapsto i_{0}(\Sigma)=(X^{0},X^{1},X^{2},X^{3})$) in the Minkowski space. Let $\tau:=-\langle X,T_{0}\rangle$ be the time function of the embedding restricted to $\Sigma$, where $T_{0}$ is the canonical time-like unit vector field of $\mathbb{R}^{1,3}$. Consider the projection $\widehat{i}_{0}(\widehat{\Sigma})$ onto the complement of $T_{0}$ i.e., onto a static slice $\mathbb{R}^{3}$. This projected surface has the metric $ \widehat{\sigma}:=\sigma+d\tau\otimes \tau$. Let $k_{0}$ be the mean curvature of $\widehat{i}_{0}( \widehat{\Sigma})$ in $\mathbb{R}^{3}$. The first reduction formula that is obtained is the following (we present a proof of this in theorem \ref{mean1}) 
\begin{eqnarray}
\label{eq:red1}
\int_{\widehat{i}_{0}( \widehat{\Sigma})}\widehat{k}_{0}\mu_{\widehat{\sigma}}= \int_{\Sigma}\left(-\sqrt{1+|\nabla\tau|^{2}_{\sigma}}\langle \mathbf{H}_{0},e_{30}\rangle-\langle\nabla[\eta]_{\nabla\tau}e_{30},e_{40}\rangle\right)\mu_{\sigma}.    
\end{eqnarray}
The second reduction formula is a bit more involved and we present the proof in theorem \ref{mean2}. In this case, one considers the surface $\Sigma$, the ball $\Omega$ that $ \Sigma$ bounds and considers the product spacetime domain $\Omega\times \mathbb{R}$. Then one solves for the so-called Jang's equation (the first progress in solving the Jang's equation was made by Schoen-Yau \cite{schoen1979proof,schoen1981proof} in their seminal proof of positive mass theorem-a context that is slightly different from the current one-see \cite{yau} for the complete detail regarding the solvability of Jang's equation). Let us briefly discuss this procedure. Let $(\Omega,g)$ be a Riemannian $3-$ manifold in the spacetime $(M,\widehat{g})$ and let $p_{ij}$ be its second fundamental form with respect to a future-directed time-like normal $e_{4}$ in a space-time
$(M,\widehat{g})$, we consider the Riemannian product $\Omega\times\mathbb{R}$ and extend $p_{ij}$ by parallel translation
along the $\mathbb{R}$ direction to a symmetric tensor $P(\cdot,\cdot)$ on $\widehat{\Omega}\times \mathbb{R}$. Such an extension makes
$P(\cdot,v)=0$, where $v$ denotes the downward unit vector in the $\mathbb{R}$ direction. Given a function $f$ on $\Omega$ with the boundary value $f=\tau$ on $ \Sigma$, one tries to find a hypersurface $\widehat{\Omega}$ as a graph of the function $f$ over $\Omega$ that has the same mean curvature as the trace of restriction of $P$ onto $\widehat{\Omega}$. One solves for the following equation on $\Omega$ in order to achieve such hypersurface $\widehat{\Omega}$
\begin{eqnarray}
\label{eq:Jang2}
\sum_{i,j=1}^{3}\left(g^{ij}-\frac{f^{i}f^{j}}{1+|\widetilde{\nabla} f|^{2}}\right)\left(\frac{\widetilde{\nabla}_{i}\partial_{j}f}{(1+|\widetilde{\nabla} f|^{2})^{\frac{1}{2}}}-P_{ij}\right)=0~on~\widehat{\Omega},\\
f=\tau~on~ \widehat{\Sigma},
\end{eqnarray}
where $\widetilde{\nabla}$ is the $g-$ connection on $\Omega$ and $f^{i}:=g^{ij}\widetilde{\nabla}_{j}$. If one can solve this Jang's equation, then indeed such a hypersurface $\widehat{\Omega}$ exists. One can compute that the metric induced on the boundary $ \widehat{\Sigma}=\partial \widehat{\Omega}$ of such a hypersurface is $ \sigma+d\tau\otimes d\tau$ i.e., the boundary $ \widehat{\Sigma}=\partial \widehat{\Omega}$ is isometric to the projection $\widehat{i}_{0}( \widehat{\Sigma})$ of the isometric image of $\Sigma$ in $\mathbb{R}^{1,3}$ onto $T_{0}$ complement. Also let $\widetilde{k}$ be the mean curvature of $ \widehat{\Sigma}$ in $\widehat{\Omega}$. First, pick an orthonormal basis $\{\widetilde{e}_{\alpha}\}_{\alpha=1}^{4}$ for the tangent space of $\Omega\times \mathbb{R}$ along $\widehat{\Omega}$ such that $\{\widetilde{e}_{i}\}_{i=1}^{3}$ is tangent to $\widehat{\Omega}$ and $\widetilde{e}_{4}$ is the downward unit normal to $\widehat{\Omega}$. The contribution to the Wang-Yau energy from the physical spacetime (i.e., term $II$ in definition \ref{wy14}) can be expressed in terms of the integral over $ \widehat{\Sigma}$ (we prove this in theorem \ref{mean2})
\begin{eqnarray}
\label{eq:red2}
    \int_{ \Sigma}\left(-\sqrt{1+|\nabla\tau|^{2}}\langle \mathbf{H},e^{'}_{3}\rangle-\alpha_{e^{'}_{3}}(\nabla\tau) \right)\mu_{ \Sigma} =\int_{ \widehat{\Sigma}}\left(\widetilde{k}-\langle\widetilde{\nabla}_{\widetilde{e}_{4}}\widetilde{e}_{4},\widetilde{e}_{3}\rangle+P(\widetilde{e}_{4},\widetilde{e}_{3})\right)\mu_{\widehat{\sigma}}.
\end{eqnarray}
Therefore, by the help of \ref{eq:red1} and \ref{eq:red2}, we may alternatively write the reduced Wang-Yau energy as the following 
\begin{eqnarray}
\label{eq:reduced}
8\pi M^{wy}_{\tau}:=\int_{ \widehat{\Sigma}}k_{0}\mu_{\widehat{\sigma}}-\int_{ \widehat{\Sigma}}\left(\widetilde{k}-\langle\widetilde{\nabla}_{\widetilde{e}_{4}}\widetilde{e}_{4},\widetilde{e}_{3}\rangle+P(\widetilde{e}_{4},\widetilde{e}_{3})\right)\mu_{\widehat{\sigma}}.    
\end{eqnarray}
We will prove that the reduced Wang-Yau energy $M^{WY}_{\tau}$ is non-negative for every admissible $\tau$ (see definition \ref{admi}) and therefore $M^{WY}=\inf_{\tau}M^{WY}_{\tau}\geq 0$ where infimum is taken over the set of admissible $\tau$.

\begin{remark}
If the Gauss curvature of $(\Sigma, \sigma)$ is positive then an embedding into $\mathbb{R}^{3}$ with $\tau=0$ is admissible. In such case, $M^{WY}$ reduces to $M^{LY}$.    In such case $(\Omega,\Sigma)$ coincide with $(\widehat{\Omega},\widehat{\Sigma})$. 
\end{remark} 

\noindent One of the important questions that arise in the context of quasi-local mass is the proof of its non-negativity and rigidity properties. \cite{shi2002positive} provided a proof of the positivity of the Brown-York mass under the assumption of the non-negativity of the scalar curvature of the bulk $\Omega$ with strictly convex boundary  $\Sigma:=\Omega$. They analyzed a non-linear parabolic equation that arises due to the use of Bartnik's gluing and extension construction \cite{bartnik1993quasi}. This entails suitable infinite asymptotically flat extension of $\Omega$ with non-negative scalar curvature by deforming the exterior of the embedding of the boundary $ i_{0}(\Sigma) \subset \mathbb{R}^{3}$ of $\Omega$ to have the same mean curvature as $\Sigma$ in $\Omega$ and to glue it to the manifold $\Omega$ along its boundary $\Sigma$. Through this construction, they extend $\Sigma$ along the radial direction to the spacelike infinity and apply the Schoen-Yau positive mass theorem \cite{schoen1979proof,schoen1981proof} and non-increasing property of the Brown-York mass as the radius of the boundary $\Sigma$ is increased. Later Liu-Yau \cite{liu2003positivity,liu2006positivity} and Wang-Yau \cite{yau} used this construction in more general settings to prove the non-negativity of Liu-Yau and Wang-Yau mass. One natural question that arises is the following: can one prove the non-negativity of the quasi-local mass by using data on the compact hypersurface $\Omega$ and its boundary $\Sigma$ in the spacetime without the extension construction of \cite{shi2002positive} and use of positive mass theorem. Richard Schoen \cite{schoen1} suggested the following problem \textit{find a quasi-local proof of nonnegativity for the Brown-York and related quasi-local masses}.

\noindent Very recently, Montiel \cite{montiel2022compact} provided a positive solution to this problem for the Brown-York mass associated with a surface $\Sigma$ under the assumption that the bulk $\Omega$ (such that $\partial\Omega= \Sigma)$ has non-negative scalar curvature and the boundary $\Sigma$ has strictly positive mean curvature using Spinor method. Unfortunately, such proof turns out to be wrong because the treatment of the boundary value problem associated with the Dirac equation fails. Here we give spinorial proof of the non-negativity of the Wang-Yau quasi-local mass assuming the dominant energy condition on the physical spacetime. As a special case, the non-negativity of the Liu-Yau mass follows (same for Brown-York too). This essentially completes the quasi-local proof of the non-negativity of the quasi-local masses. We employ the spinor technique with two different boundary conditions namely APS and MIT Bag boundary conditions. A technical assumption on the spinor field being used is required to obtain the correct proof in the case of the APS boundary condition. Contrary to the APS boundary condition, the MIT bag boundary condition is much better suited to handle the positivity question.  Below we describe the main theorem and idea of the proof.

Let us concretely define the surfaces that we will be dealing with in this article.  
\begin{definition}[Surfaces of Interest]
\label{surfaces}
$ \Sigma$ is a topological $2-$sphere in the physical spacetime $(M,\widehat{g})$. $\Sigma$ bounds a spacelike topological ball $\Omega$. The induced metric on $\Sigma$ is $ \sigma$ and on $\Omega$ is $g$. $i_{0}: \Sigma\hookrightarrow \mathbb{R}^{1,3}$ is the isometric embedding of $\Sigma$ in $\mathbb{R}^{1,3}$. The image is denoted by $i_{0}(\Sigma)$. Let $T_{0}$ be the canonical unit timeline vector field of the Minkowski space. $\widehat{i}_{0}(\widehat{\Sigma})$ is defined as the image of the projection of $i_{0}(\Sigma)$ onto the complement of $T_{0}$ i.e., onto $\mathbb{R}^{3}$. 

Now consider the Riemannian product $\Omega\times \mathbb{R}$ in the spacetime $(M,\widehat{g})$. Let $\widehat{\Omega}$ be the graph of the function $f$ solving the Jang's equation \ref{eq:Jang2} over $\Omega$ with boundary condition $f=\tau$ on $\Sigma$. $ \widehat{\Sigma}$ is defined to be the boundary of $\widehat{\Omega}$. $ \widehat{\Sigma}$ is isometric to $\widehat{i}_{0}(\widehat{\Sigma})$. In particular this isometric map is denoted by $\widehat{i}_{0}$. 
\end{definition}

\noindent We denote the spin bundle on the spacelike topological $3-$ball $\widehat{\Omega}$ by $S^{\widehat{\Omega}}$ and let $ \widehat{\Sigma}$ be its boundary $\partial \widehat{\Omega}$. Let $(e_{1},e_{2},e_{3})$ be an orthonormal basis of the frame bundle of $(\widehat{\Omega},\widetilde{g})$. Let $\text{Cl}(\widehat{\Omega},\mathbb{C})$ denote the clifford algebra. The spin representation gives a surjective algebra homomorphism $\rho:\text{Cl}(\widehat{\Omega},\mathbb{C})\to \text{End}(\Delta)$, where $\Delta$ is the complex vector space of dimension $2$ equipped with a Hermitian inner product $\langle\cdot,\cdot\rangle$ induced by the Riemannian structure on $\widehat{\Omega}$. As a bundle $S^{\widehat{\Omega}}$ maybe identified with the trivial bundle $\widehat{\Omega}\times \Delta\to \Delta$. Each tangent vector $Y\in T_{x}\widehat{\Omega}$ induces a skew-adjoint map $\rho(Y)\in \text{End}(S^{\widehat{\Omega}}_{x})$. For the orthonormal basis $\{e_{i}\}_{i=1}^{3}$, Clifford relation holds 
\begin{eqnarray}
\label{eq:cliff}
\rho(e_{i})\rho(e_{j})+\rho(e_{j})\rho(e_{i})=-2\delta_{ij}\text{id}    
\end{eqnarray}
for $i,j=1,2,3$. On the boundary $ \widehat{\Sigma}$, there is a spinor bundle $S^{ \widehat{\Sigma}}$. 

Naturally, restriction of a section $s$ to the boundary $ \widehat{\Sigma}$ is a section of the bundle $S^{ \widehat{\Sigma}}$. We choose $e_{3}$ to be the interior pointing normal field to the boundary $ \widehat{\Sigma}$. We next consider the spin connection with respect to the metric $g$. The spin connection
is a connection $\widehat{\nabla}$ on $S^{\widehat{\Omega}}$ which is compatible with Clifford multiplication and
which is compatible with the Hermitian inner product on $S^{\widehat{\Omega}}$. Let $\Psi$ be a section of the spinor bundle $S^{\widehat{\Omega}}$. We want to solve the Dirac equation in $\widehat{\Omega}$ with appropriate boundary conditions on the boundary $\partial\widehat{\Omega}$. We denote the restriction of $\Psi$ onto $\partial\widehat{\Omega}$ by $\Psi_{ \widehat{\Sigma}}$ which is naturally identified as a section of the bundle spin bundle $S^{\widehat{\Omega}}$. Let $\mathcal{S}^{APS}$ and $\mathcal{S}^{MIT}$ denote the space of spinors that solve the Dirac equation with prescribed $APS$ and $MIT$ Bag boundary conditions, respectively. Let $S^{ \widehat{\Sigma}}$ be the space of spinors on $\widehat{\Omega}$ and $\Gamma^{\geq 0}(S^{ \widehat{\Sigma}})$ denote the space of spinors on $ \widehat{\Sigma}$ that consists of positive frequency spinors. Define $P_{\geq 0}$ to be the projection operator $P_{\geq 0}: S^{\widehat{\Sigma}}\to \Gamma^{\geq 0}(S^{ \widehat{\Sigma}})$ which is a pseudo-differential operator of degree zero. In addition, consider the chirality subbundles ${S^{\widehat{\Sigma}}}_{\pm}$ and projections $\Pi_{\pm}:S^{\widehat{\Sigma}}\to {S^{\widehat{\Sigma}}}_{\pm}$. MIT Bag boundary condition is essentially fixing the $\Pi_{+}$ or $\Pi_{-}$ component of the boundary restriction of $\Psi$. Similar notations are used for the surface $\Sigma$ and the ball $\widehat{\Omega}$ that it bounds. Let us define the following two spaces
\begin{definition}[Spinor spaces of interest]
\label{spinor}
Let $C^{\infty}(S^{ \widehat{\Sigma}})$ be the space of $C^{\infty}$ spinors on $ \widehat{\Sigma}$ and $\Gamma^{\geq 0}(S^{ \widehat{\Sigma}})$ denote the space of $C^{\infty}$ spinors on $ \widehat{\Sigma}$ that consists of positive frequency eigenspinors. Define $P_{\geq 0}$ to be the $L^{2}$ projection operator $P_{\geq 0}: S^{ \widehat{\Sigma}}\to \Gamma^{\geq 0}(S^{ \widehat{\Sigma}})$. We define
\begin{eqnarray}
\mathcal{S}^{MIT}:=\left\{\Psi\in C^{\infty}(S^{\widehat{\Omega}})|\widehat{D}\Psi=0~on~\widehat{\Omega},~\Pi_{+}(\Psi)=\Pi_{+}(\alpha),\alpha\in C^{\infty}(S^{ \widehat{\Sigma}})\right\},\\    
 \mathcal{S}^{APS}:=\left\{\Psi\in C^{\infty}(S^{\widehat{\Omega}})|\widehat{D}\Psi=0~on~\widehat{\Omega},~P_{\geq 0}\Psi=P_{\geq 0}\alpha,~\alpha\in C^{\infty}(S^{ \widehat{\Sigma}}),  |P_{\geq 0}\Psi_{ \widehat{\Sigma}}|\leq |\Psi_{ \widehat{\Sigma}}|\right\}
\end{eqnarray}     
\end{definition}

\noindent We are interested in proving the positivity of the following two entities. For now we assume that $\mathcal{S}^{APS}$ and $\mathcal{S}^{MIT}$ are non-empty. Proving that these are non-empty is a major task that is to be executed in section \ref{spinorproof} and it depends crucially on the curvature condition.

\begin{definition}[Spinor based quasi-local entity]
Assume that the set $\mathcal{S}^{MIT}$ is non-empty. The spinor based quasi-local entity is defined for $\Psi\in \mathcal{S}^{MIT}$ as follows
\begin{eqnarray}
8\pi M^{\Psi}:=\int_{ \widehat{\Sigma}}\left(k_{0}-(\widetilde{k}-\langle\widetilde{\nabla}_{\widetilde{e}_{4}}\widetilde{e}_{4},\widetilde{e}_{3}\rangle+P(\widetilde{e}_{4},\widetilde{e}_{3}))\right)|\Psi|^{2}\mu_{\widehat{\sigma}}
\end{eqnarray}    
\end{definition}
 
\begin{remark}
Note that $\inf_{\Psi\in \mathcal{S}^{MIT},|\Psi|=1} M^{\Psi}=M^{WY}_{\tau}$ for any admissible $\tau$ (definition \ref{admi}). Therefore positivity of $M^{\Psi}$ for any $\Psi$ solving Dirac equation with MIT Bag boundary condition implies the positivity of the Wang-Yau quasi-local energy.
\end{remark}

\noindent The main theorem that we will prove in this article is the following.  
\begin{theorem}
\label{1} Let $(\Sigma, \sigma)$ be a spacelike topological $2-$sphere that bounds a spacelike topological ball $\Omega$ in the physical spacetime $(M,\widehat{g})$ verifying dominant energy condition \ref{eq:dominant}. Moreover, assume $\Sigma$ is not an apparent horizon in $(M,\widehat{g})$, and its mean curvature vector $\mathbf{H}$ is spacelike. Let $i_{0}:\Sigma\to \mathbb{R}^{1,3}$ be an isometric embedding into the Minkowski space and let $\tau$ denote the restriction of the time function $t$ on $i_{0}(\Sigma)$. Let $\Sigma,\Omega, \widehat{\Sigma}$, and $\widehat{\Omega}$ be as in definition \ref{surfaces}. Assume $\tau$ is admissible (see definition \ref{admi}) and $\mathcal{S}^{MIT}$ is non-empty then for a $\Psi\in \mathcal{S}^{MIT}$,
(a) \begin{eqnarray}
M^{\Psi}=\frac{1}{8\pi}\int_{ \widehat{\Sigma}}\left(k_{0}-(\widetilde{k}-\langle\widetilde{\nabla}_{\widetilde{e}_{4}}\widetilde{e}_{4},\widetilde{e}_{3}\rangle+P(\widetilde{e}_{4},\widetilde{e}_{3}))\right)|\Psi|^{2}\mu_{\widehat{\sigma}}   
\end{eqnarray}
is non-negative. Moreover $M^{\Psi}$ vanishes if and only if the spacetime $(M,\widehat{g})$ is Minkowski along $ \widehat{\Sigma}$,
(b) If $\mathcal{S}^{APS}$ is non-empty, then $M^{WY}$ is non-negative and $M^{WY}$ vanishes if and only if the spacetime $(M,\widehat{g})$ is Minkowski along $ \widehat{\Sigma}$. 
\end{theorem}

Using a similar technique we are able to provide quasi-local proofs of the positivity of the Brown-York and Li-Yau masses. These are special cases when the embedding is chosen to be into $\mathbb{R}^{3}$ instead of the Minkowski space.

\begin{theorem}
\label{2}
Let $\Sigma$ be a topological $2-$sphere in the spacetime $(M,\widehat{g})$ bounding a spacelike topological ball $\Omega$. Let the scalar curvature of $\Omega$ be non-negative.
Consider isometric embedding of $(\Sigma, \sigma)$ into the Euclidean space $(\mathbb{R}^{3},\delta)$ $i:\Sigma\hookrightarrow \mathbb{R}^{3}$. Let $\widehat{k}$ be the mean curvature of $\Sigma$ while viewed as an embedded surface in $\Omega$ in the physical spacetime $(M,\widehat{g})$.
  Let $k_{0}$ be the same for its image $i(\Sigma)$ in $\mathbb{R}^{3}$. Also assume that the Gaussian curvature of $(\Sigma,\sigma)$ is strictly positive. Then (a) if $\mathcal{S}^{APS}$ is non-empty, the Brown-York mass 
\begin{eqnarray}
 M^{BY}=\frac{1}{8\pi}\int_{\Sigma}(k_{0}-\widehat{k})\mu_{\sigma}.   
\end{eqnarray}   
is non-negative, vanishes if and only if $(\Sigma,\sigma)$ is a surface in the Euclidean space $(\mathbb{R}^{3},\delta)$, 
(b) If $\mathcal{S}^{MIT}$ is non-empty, then for a $\Psi\in \mathcal{S}^{MIT}$
\begin{eqnarray}
   M^{\Psi}_{BY}:=\frac{1}{8\pi}\int_{\Sigma}(k_{0}-\widehat{k})|\Psi|^{2}\mu_{\sigma} \geq 0
\end{eqnarray}
and it vanishes if and only if the spacetime $(\Sigma, \sigma)$ is a surface in the Euclidean space $(\mathbb{R}^{3},\delta)$.
\end{theorem}
\begin{remark}
Note that $\inf_{\Psi\in \mathcal{S}^{MIT},|\Psi|=1}M^{BY}_{\Psi}=M^{BY}\geq 0$ as a consequence of theorem \ref{2}.    
\end{remark}

\subsection{Main difficulty in the compact case and the idea of the proof}
\noindent Following the development of the spinor ideas \cite{witten1981new,parker1982witten} in proving the positive mass theorem, the strategy is formally straightforward. One would expect an appropriate application of Bochner-type identity and the curvature condition to yield the positivity of the quasi-local mass. However, Wang-Yau quasi-local mass is much more involved than the definition of ADM mass or even other quasi-local masses such as Brown-York \cite{brown1993quasilocal} or Li-Yau \cite{liu2003positivity}. In fact, two non-trivial theorems need to be proved before arriving at the step where one can start thinking about applying Bochner-type identity. We should mention that there is proof of the positivity of Brown-York mass using spinor by Montiel \cite{montiel2022compact}. Unfortunately, there is a fatal flaw in the proof. We shall discuss this in this section. One of the main difficulties with using a spinor in the compact case is solving the boundary value problem associated with the Dirac equation and constructing a unit norm spinor.

 Let us recall that $\widehat{\Omega}$ is a spacelike topological ball in the spacetime $(M,\widehat{g})$ that  has a smooth boundary $\partial\widehat{\Omega}$ and denote this boundary by $ \widehat{\Sigma}$. The main idea behind the proof is twofold. First, we want to make use of a Bochner formula for the spinor $\psi\in C^{\infty}(S^{\widehat{\Omega}})$. We denote the restriction of the spinor $\psi$ to the boundary $ \widehat{\Sigma}$ by $\psi$ as well. The Bochner type formula in use reads 
\begin{eqnarray}
\label{eq:10}
 \int_{ \widehat{\Sigma}}\left(\langle D\psi,\psi\rangle-\frac{1}{2}H|\psi|^{2}\right)\mu_{\widehat{\sigma}}= \frac{1}{4}\int_{\widehat{\Omega}}R^{\widehat{\Omega}}|\psi|^{2}\mu_{\widehat{\Omega}}-\int_{\widehat{\Omega}}|\widehat{D}\psi|^{2}\mu_{\widehat{\Omega}}+\int_{\widehat{\Omega}}|\hnabla\psi|^{2}\mu_{\widehat{\Omega}}.     
\end{eqnarray}
Here $\widehat{\Omega}$ is a topological ball embedded in a spacelike slice and $ \widehat{\Sigma}$ is its boundary. The term $-\int_{\widehat{\Omega}}|\widehat{D}\psi|^{2}\mu_{\widehat{\Omega}}$ is problematic since one wants to employ a positivity condition on the scalar curvature $R^{\widehat{\Omega}}$ (technically a weak positivity condition to be precise in the context of Wang-Yau quasi-local energy and dominant energy condition assumption on the spacetime). Therefore, one wishes to work with a Dirac spinor, that is, solve for $\psi$ that verifies 
\begin{eqnarray}
 \widehat{D}\psi=0 ~on~\widehat{\Omega}.   
\end{eqnarray}
However, one needs an appropriate boundary condition on $ \widehat{\Sigma}$ to solve this equation. Clearly, the Dirichlet boundary condition ends up being too strong and does not work (think of a first-order ordinary differential equation on a closed interval; clearly imposing boundary conditions on both ends may yield no solution at all or think of the holomorphic functions $\varphi$ that verify $\overline{\partial}\varphi=0$). Motivated by the Dirichlet boundary value problem from second-order elliptic equations, heuristically, one is required to impose ``half" of the Dirichlet boundary condition. One could choose different boundary conditions to solve the Dirac equation. Here we choose two different boundary conditions namely Atyiah-Patodi-Singer (APS) boundary condition \cite{atiyah1975spectral1,atiyah1975spectral2,atiyah1975spectral3} and the MIT Bag boundary condition \cite{johnson1975bag} as it serves our purpose the best. Thankfully, much is developed by \cite{bar} in the context of the boundary value problem associated with the Dirac operator. To prove the existence of a unique solution given either of the above two boundary conditions, one needs certain positivity conditions on the scalar curvature and the mean extrinsic curvature of the boundary $ \widehat{\Sigma}$. But this is still not enough as the expression of quasi-local energy is close to a form on the left-hand side of \ref{eq:10} with $|\psi|=1$ on $ \widehat{\Sigma}$. Therefore, in addition, we need to prove the existence of a solution of the Dirac equation whose norm is constant on the boundary. A parallel spinor is a constant norm spinor. But, the existence of a parallel spinor puts Ricci flatness restriction on $\widehat{\Omega}$. In the positivity proof of ADM mass by \cite{parker1982witten}, the geometry of the Cauchy slice is asymptotically Euclidean. Euclidean space admits parallel spinor and therefore it was natural to solve the boundary value problem associated with Dirac equation with 
\begin{eqnarray}
 \psi\to 1,~r\to\infty   
\end{eqnarray}
where $r$ denotes the geodesic distance from a reference point in the interior. In the current context, however, a different approach is necessary. The Wang-Yau quasi-local mass is defined with respect to the Minkowski background. There exist parallel spinors on the Minkowski space and therefore, one can pull back a parallel spinor by a suitably constructed isometry map. This step is non-trivial and needs careful attention. 

Secondly, in the proof of the positive mass theorem or positivity of other quasi-local masses, the left-hand side of \ref{eq:10} usually corresponds to the mass after setting $|\psi|=1$. However, as mentioned earlier, the definition of Wang-Yau quasi-local energy is much more involved and it is not obvious at all that the expression in \ref{eq:wy1} can be reduced to a form on the left-hand side of \ref{eq:10} modulo an additional term. This is the second non-trivial step in the proof and requires a series of calculations involving careful manipulation of the structure equations of the Cauchy slice in spacetime (see section \ref{reduction}). This additional difficulty arises because Wang-Yau quasi-local energy is essentially defined for a spacetime rather than for a spacelike slice. Therefore the second fundamental form of a Cauchy slice enters into the picture. This reduction procedure requires solving the so-called Jang's equation. Proving the existence of a bounded solution to Jang's equation is very delicate. The sufficient condition for a bounded solution to exist is that the topological sphere for which the mass is defined is not an apparent horizon (it is difficult to define the mass for an apparent horizon precisely due to the fact that the solution of Jang's equation blows up). This imposes a condition on the mean curvature of $ \widehat{\Sigma}$ and the second fundamental form of $\widehat{\Omega}$ in the spacetime.

In the case of Brown York mass, apart from the common difficulty associated with the solution of the boundary value problem, the rest are straightforward. In particular, the construction of the isometric embedding into $\mathbb{R}^{3}$ is trivial upon imposing a point-wise positivity of the Gauss curvature of the topological $2-$sphere and an application of Nirengerg-Pogolerov theorem \cite{pogorelov1952regularity}. In addition, the brown York mass is explicitly defined in terms of the embedding into flat $\mathbb{R}^{3}$ instead of Minkowski space i.e., 
\begin{eqnarray}
\label{eq:BY}
M_{Brown-York}:=\frac{1}{8\pi}\int_{ \widehat{\Sigma}}(k_{0}-\widehat{k})\mu_{\widehat{\sigma}}  
\end{eqnarray}
where $\widehat{k}$ is the mean curvature of $ \widehat{\Sigma}$ while viewed as an embedded surface in a Cauchy slice in the physical spacetime and $k_{0}$ is the mean curvature of its isometric image in the Euclidean space $\mathbb{R}^{3}$. Compare this to the expression \ref{eq:wy2} of the Wang-Yau quasi-local energy. Even though we will reduce \ref{eq:wy2} to $8\pi M^{\Psi}$, this step is highly non-trivial and we discuss this in section \ref{reduction}. Nevertheless, we need to solve the Dirac equation with the prescribed boundary condition which is necessary to prove the positivity of any quasi-local mass. First, let us describe the $MIT$ bag boundary condition. 

Consider the Riemannian manifold $(\widehat{\Omega},\widetilde{g})$ diffeomorphic to a $3-$ball with $C^{\infty}$ boundary $ \widehat{\Sigma}$. Let $(e_{1},e_{2},e_{3})$ be a $\widetilde{g}-$orthonormal frame of $\widehat{\Omega}$. Let $e_{3}$ be the unit inward normal to $ \widehat{\Sigma}$. Define two point-wise projections 
\begin{eqnarray}
 \Pi_{\pm}:S^{ \widehat{\Sigma}}\to S^{ \widehat{\Sigma}}\\
 \psi\mapsto \Pi_{\pm}(\psi):=\frac{1}{2}(\psi\pm \sqrt{-1}\rho(e_{3})\Psi).
\end{eqnarray}
Since the dimension of $ \widehat{\Sigma}$ is even, this is nothing but the projection onto the $\pm$ chirality subbundles $S^{ \widehat{\Sigma}}_{\pm}$ of $S^{ \widehat{\Sigma}}$. $\Pi_{+}$ and $\Pi_{-}$ are self-adjoint and orthogonal at each $S^{ \widehat{\Sigma}}_{x}$. Let $\Psi\in H^{s}(S^{\widehat{\Omega}})\cap H^{s-\frac{1}{2}}(S^{ \widehat{\Sigma}}),~s\geq 1$, where $H^{s}$ here is defined to the Sobolev space of sections of spin bundles that is the completion of $C^{\infty}$ with respect to the equivalent norm $||\Psi||_{L^{2}(\widehat{\Omega}}+\sum_{I=1}^{s}||\widehat{\nabla}^{I}\Psi||_{L^{2}(\widehat{\Omega})}$. The following boundary value problem is solved
\begin{eqnarray}
\widehat{D}\Psi=0~on~\widehat{\Omega}\\
\Pi_{+}\Psi=\Pi_{+}\alpha~on~ \widehat{\Sigma}
\end{eqnarray}
for $\alpha\in S^{ \widehat{\Sigma}}$. We shall prove in section \ref{MITbc} that this boundary value problem has a unique solution. The next thing is to use $\Psi$ in the Bochner identity \ref{eq:10} and set $\widehat{D}\Psi=0$. After a series of delicate estimates, we obtain an inequality of the following type \begin{eqnarray}
M_{\Psi}:=\frac{1}{8\pi}\int_{ \widehat{\Sigma}}\left(k_{0}-(\widetilde{k}-\langle X,e_{3}\rangle)\right)|\Psi|^{2}\mu_{\widehat{\sigma}}\geq 0  
\end{eqnarray}  
under the condition that $R^{\widehat{\Omega}}\geq 2|X|^{2}-2\text{div}X$ where $X$ is a $\widehat{\Omega}$ parallel vector field dependent on the second fundamental form of $\widehat{\Omega}$ in $(M,\widehat{g})$ (see section \ref{MITbc} for the details).  This is done in section \ref{reduction} via a series of calculations and requires solving Jang's equation. This in turn requires a positivity condition on the generalized mean curvature $(\widetilde{k}-\langle X,e_{3}\rangle)$ of $ \widehat{\Sigma}$ (see section \ref{reduction}). If one solves the Dirac equations in the class $|\Psi|=1$, then non-negativity of Wang-Yau energy $M^{WY}_{\tau}=\inf_{\Psi\in \mathcal{S}^{MIT},|\Psi|=1}M^{\Psi}$ follows from the non-negativity of $M^{\Psi}$. To summarize, the main steps in proving the non-negativity of Wang-Yau energy are proving the Fredholm property of the Dirac operator under $MIT$ bag boundary condition and appropriate scalar curvature condition, reduction of the Wang-Yau energy functional to the form \ref{eq:reduced}, and obtaining the non-negativity of $\widetilde{M}^{\Psi}$. 

Now we proceed to discuss the solution with a non-local boundary condition namely APS boundary condition. On the boundary $ \widehat{\Sigma}$, $L^{2}(S^{ \widehat{\Sigma}})$ splits into two orthogonal subspaces. This is as follows. Let $D$ be the Dirac operator on $ \widehat{\Sigma}$. Consider the eigenvalue equation for $D$
\begin{eqnarray}
 D\psi=\lambda \psi,~\lambda\in \mathbb{R}.   
\end{eqnarray}
$ \widehat{\Sigma}$ is closed and $D$ is symmetric on $ \widehat{\Sigma}$. The spectrum $\{\lambda\in \mathbb{R}\}$ is symmetric with respect to zero. A generic section of $S^{ \widehat{\Sigma}}$ is written as $\sum_{i}a_{i}\psi_{i},~a_{i}\in \mathbb{C}$ and $D\psi_{i}=\lambda_{i}\psi_{i},\lambda_{i}\in \mathbb{R}$. One can split $L^{2}(S^{ \widehat{\Sigma}})$ into $L^{2}_{+}(S^{ \widehat{\Sigma}})$ and $L^{2}_{-}(S^{ \widehat{\Sigma}})$ as follows 
\begin{eqnarray}
L^{2}_{+}(S^{ \widehat{\Sigma}}):=\left\{\psi\in L^{2}(S^{ \widehat{\Sigma}})|\psi=\sum_{i=1}a_{i}\psi_{i},~D\psi_{i}=\lambda_{i}\psi_{i},\lambda_{i}\geq 0\right\},\\
L^{2}_{-}(S^{ \widehat{\Sigma}}):=\left\{\psi\in L^{2}(S^{ \widehat{\Sigma}})|\psi=\sum_{i=1}a_{i}\psi_{i},~D\psi_{i}=\lambda_{i}\psi_{i},\lambda_{i}< 0\right\}. 
\end{eqnarray}
Naturally $L^{2}_{+}(S^{ \widehat{\Sigma}})$ and $L^{2}_{-}(S^{ \widehat{\Sigma}})$ are $L^{2}$ orthogonal.
Let us denote by $P_{\geq 0},P_{<0}$ the projection operators
\begin{eqnarray}
 P_{\geq 0}: L^{2}(S^{ \widehat{\Sigma}})\to L^{2}_{+}(S^{ \widehat{\Sigma}}),\\
 P_{< 0}: L^{2}(S^{ \widehat{\Sigma}})\to L^{2}_{-}(S^{ \widehat{\Sigma}}).
\end{eqnarray}
The APS boundary value problem for the Dirac equation reads 
\begin{eqnarray}
 \widehat{D}\Psi=0~on~\widehat{\Omega}\\
 P_{\geq 0}\Psi=P_{\geq 0}\alpha~on~ \widehat{\Sigma}
\end{eqnarray}
for $\alpha\in L^{2}(S^{ \widehat{\Sigma}})$. We prove the existence of a unique solution to this problem under the curvature condition $R^{\widehat{\Omega}}\geq 2|X|^{2}-2\text{div}X$ and the mean curvature condition $\widetilde{k}-\langle X,e_{3}\rangle>0$ for an appropriate choice of $X$ (see section \ref{proof1} for the definition of $X$). Note that to prove the isomorphism property of the Dirac operator between relevant Sobolev spaces, we did not need the condition $\widetilde{k}-\langle X,e_{3}\rangle>0$ in the case of MIT Bag boundary condition rather needed for the solvability of the Jang's equation. If $\Psi\in \mathcal{S}^{APS}$ (defined in \ref{spinor}), then we obtain a spin inequality of the following type 
\begin{eqnarray}
\frac{1}{8\pi}\int_{ \widehat{\Sigma}}\left(k_{0}-(\widetilde{k}-\langle X,e_{3}\rangle)\right)|P_{\geq 0}\Psi|^{2}\mu_{\widehat{\sigma}}\geq 0.    
\end{eqnarray}
Moreover, we prove that the equality holds if $ \widehat{\Sigma}$ is embedded in $\mathbb{R}^{3}$ and $\Psi$ is parallel (see section \ref{dirac}). Let's call this spinor at the case of equality $\Phi$. In such case we also observe that $\Phi_{ \widehat{\Sigma}}=P_{\geq 0}\Phi_{ \widehat{\Sigma}}$. In particular, we can set $|P_{\geq 0}\Phi|^{2}=1$ without loss of generality on $ \widehat{\Sigma}$ in $\mathbb{R}^{3}$ when the equality holds. Now we pull this spinor $P_{\geq 0}\Phi$ back to the surface $ \widehat{\Sigma}$ in $\widehat{\Omega}$ by an isometry and use this as the boundary condition for the Dirac equation on $\widehat{\Omega}$ i.e., we solve the problem 
\begin{eqnarray}
\widehat{D}\Psi=0~on~\widehat{\Omega}\\
 P_{\geq 0}\Psi=P_{\geq 0}\alpha~on~ \widehat{\Sigma}    
\end{eqnarray}
with $\alpha=P_{\geq 0}\Phi$ (pulled back to be precise). This equation is solvable as proved in section \ref{dirac}. Therefore we have the inequality 
\begin{eqnarray}
\frac{1}{8\pi}\int_{ \widehat{\Sigma}}\left(k_{0}-(\widetilde{k}-\langle X,e_{3}\rangle)\right)|P_{\geq 0}\Phi|^{2}\mu_{\widehat{\sigma}}\geq 0,    
\end{eqnarray}
but the isometry preserves the norm of the spinor and therefore $|P_{\geq 0}\Phi|=1$ on $ \widehat{\Sigma}$ in $\widehat{\Omega}$ as well. Therefore, we end up proving 
\begin{eqnarray}
\mathcal{Q}:=\frac{1}{8\pi}\int_{ \widehat{\Sigma}}\left(k_{0}-(\widetilde{k}-\langle X,e_{3}\rangle)\right)\mu_{\widehat{\sigma}}\geq 0.  
\end{eqnarray}
The last step that remains is to prove that the Wang-Yau energy (\ref{eq:wy2}) can be reduced to the form $\mathcal{Q}$. This step is independent of the previous step i.e., use of spinor and follows in a similar way as the MIT Bag case.  

Montiel \cite{montiel2022compact} provided a proof of the positivity of Brown-York mass using spinor technique. However, there is a gap in his treatment of the boundary value problem that we discuss now. He uses a generalized APS boundary condition of the type 
\begin{eqnarray}
 \widehat{D}\psi=0~on~\widehat{\Omega}\\
 \psi_{\geq \lambda}=\alpha_{\geq \lambda}~on~ \widehat{\Sigma}=\partial\widehat{\Omega},
\end{eqnarray}
for $\alpha\in C^{\infty}(S^{ \widehat{\Sigma}})$.
Here $\psi_{\geq \lambda}$ denotes the part of $\psi$ with frequency greater than or equal to $\lambda\in \mathbb{R}$.
He then proves the existence of a unique solution in relevant Sobolev spaces showing the isomorphism property of the Dirac operator between respective Sobolev spaces with this boundary condition for every $\lambda>-\infty$. Therefore, he proves in essence the solvability of the Dirichlet boundary value problem for the Dirac equation on a compact domain. This is, however, impossible. The main mistake made by Montiel is his proof of surjectivity. Recall that proceeds to show first the triviality of the kernel i.e., show that the following problem does not have any solution 
\begin{eqnarray}
\widehat{D}\psi=0~on~\widehat{\Omega}\\
 \psi_{\geq \lambda}=0.    
\end{eqnarray}
This step goes through for any $\lambda\leq 0$ with appropriate curvature assumption. However, to show the associated adjoint problem has trivial solution, he solves the problem 
\begin{eqnarray}
\widehat{D}\psi=0~on~\widehat{\Omega}\\
 \psi|_{> \lambda}=0  
\end{eqnarray}
where a strict positivity of the mean curvature is used. However, the problem happens here since to solve for the co-kernel, the current adjoint problem is 
\begin{eqnarray}
\widehat{D}\psi=0~on~\widehat{\Omega}\\
 \psi|_{> -\lambda}=0  
\end{eqnarray}
which does not yield a trivial solution for $\lambda<0$. Therefore the remaining argument breaks down.

\section{Extrinsic Spin Geometry}
\noindent Let $(\widehat{\Omega},g)$ be a compact connected oriented 3 dimensional $C^{\infty}$ Riemannian manifold with boundary $\partial \widehat{\Omega}= \widehat{\Sigma}$. Let $\hnabla$ be the metric compatible connexion on $(\widehat{\Omega},g)$ and $\nabla$ is the induced connexion on $ \widehat{\Sigma}$. For two smooth vector fields $A,B\in T \widehat{\Sigma}$, $\widehat{\nabla}$ and $\nabla$ are related by 
\begin{eqnarray}
 \widehat{\nabla}_{A}B=\nabla_{A}B+\mathcal{H}(A,B)N,   
\end{eqnarray}
where $N$ is unit normal vector to the boundary $ \widehat{\Sigma}$ and $\mathcal{H}$ is the second fundamental form of $ \widehat{\Sigma}$ in $\widehat{\Omega}$. We will work with spin bundles and their sections. 
We denote the spin bundle on $\widehat{\Omega}$ by $S^{\widehat{\Omega}}$ while that of $ \widehat{\Sigma}$ is denoted by $S^{ \widehat{\Sigma}}$. In our case, $\widehat{\Omega}$ is a topological ball, and naturally, $ \widehat{\Sigma}$ is a topological $2-$sphere.
On the spin bundle $S^{\widehat{\Omega}}$ there is a natural Hermitian inner product $\langle\cdot,\cdot\rangle$ that is compatible with the connexion $\hnabla$ (we denote the Levi-Civita connexion and its spin connexion by the same symbol) and Clifford multiplication $\rho: Cl(\widehat{\Omega},\mathbb{C})\to \text{End}(S^{\widehat{\Omega}})$. We denote a spinor a section of the spin bundle by $\psi\in \Gamma(S^{\widehat{\Omega}})$. Naturally, its restriction to $ \widehat{\Sigma}$ is a section of $S^{ \widehat{\Sigma}}$. Let us consider $\{e_{i}\}_{i=1}^{3}$ as an orthonormal frame on $ \widehat{\Sigma}$ where $e_{3}=\nu$ is the interior pointing normal to the boundary $ \widehat{\Sigma}$. The following relation holds $\forall X,Y\in \Gamma(T\widehat{\Omega})$
\begin{eqnarray}
\label{eq:compat}
X\langle \psi_{1},\psi_{2}\rangle=\langle\hnabla_{X}\psi_{1},\psi_{2}\rangle+\langle\psi_{1},\hnabla_{X}\psi_{2}\rangle \\
\hnabla_{X}(\rho(Y)\psi)=\rho(\hnabla_{X}Y)\psi+\rho(Y)\hnabla_{X}Y\\
\label{eq:clif}
\langle \rho(X)\psi_{1},\psi_{2}\rangle=-\langle \psi_{1},\rho(X)\psi_{2}\rangle\\
\langle \rho(X)\psi_{1},\rho(X)\psi_{2}\rangle=g(X,X)\langle \psi_{1},\psi_{2}\rangle.
\end{eqnarray}
These are the intrinsic identities that are verified by the geometric entities on $\widehat{\Omega}$. Intrinsically, we can define the Dirac operator on $\widehat{\Omega}$ as the composition map 
\begin{eqnarray}
 S^{\widehat{\Omega}}\xrightarrow{\widehat{\nabla}} T^{*}\widehat{\Omega}\times S^{\widehat{\Omega}}\xrightarrow{\rho} S^{\widehat{\Omega}}.  
\end{eqnarray}
In local coordinates, the above decomposition gives the Dirac operator $\widehat{D}$ 
\begin{eqnarray}
 \widehat{D}:=\sum_{i}\rho(e_{i})\widehat{\nabla}_{e_{i}}.  
\end{eqnarray}
The boundary map $i: \widehat{\Sigma}\hookrightarrow\widehat{\Omega}$, induces a connexion on $ \widehat{\Sigma}$ that verifies 
\begin{eqnarray}
\hnabla_{X}Y=\nabla_{X}Y+\mathcal{H}(X,Y)\nu,    
\end{eqnarray}
where $X,Y\in \Gamma(T \widehat{\Sigma})$ and $\mathcal{H}$ is the second fundamental form of $ \widehat{\Sigma}$ while viewed as an embedded surface in $\widehat{\Omega}$. Now this split extends to the spinor bundles. One can consider the restricted bundle $S^{\widehat{\Omega}}|_{ \widehat{\Sigma}}$ which is isomorphic to $S^{ \widehat{\Sigma}}$ since dimension of $ \widehat{\Sigma}$ is even.
The spin connection splits as follows 
\begin{eqnarray}
\hnabla_{X}\psi=\nabla_{X}\psi+\frac{1}{2}\rho(\mathcal{H}(X,\cdot))\rho(\nu)\psi    
\end{eqnarray}
for $X\in \Gamma(T \widehat{\Sigma})$. With this split, we can write the hypersurface Dirac operator in terms of the bulk Dirac operator and the second fundamental form. We have the following proposition relating the Dirac operator of $\widehat{\Omega}$ and that of its boundary $ \widehat{\Sigma}$ 
\begin{proposition}
\label{dirac1}
Let $ \widehat{\Sigma}$ be the smooth boundary of the Riemannian $3-$manifold $\widehat{\Omega}$ and $H$ be the mean curvature of $ \widehat{\Sigma}$ in $\widehat{\Omega}$.  
The boundary Dirac operator $D$ on $ \widehat{\Sigma}$ induced by the Riemannain structure of $\widehat{\Omega}$ and the bulk Dirac operator $\widehat{D}$ on $\widehat{\Omega}\cup \partial\widehat{\Omega}=\widehat{\Omega}\cup \widehat{ \widehat{\Sigma}}$ verify
\begin{eqnarray}
\label{eq:hypersurface}
D\psi=-\rho(\nu)\widehat{D}\psi-\widehat{\nabla}_{\nu}\psi+\frac{1}{2}H\psi    
\end{eqnarray}    
for a $C^{\infty}$ section $\psi$ of the bundle $S^{\widehat{\Omega}}$.
\end{proposition}
\begin{proof} The proof is a result of straightforward calculations. Let $\psi\in C^{\infty}(S^{\widehat{\Omega}})$. We denote the restriction of $\psi$ to the boundary $ \widehat{\Sigma}$ by $\psi$ as well for notational convenience. The hypersurface Dirac operator is defined as follows 
\begin{eqnarray}
D\psi=\sum_{i=1}^{2}\rho(e_{i})\rho(\nu)\nabla_{e_{i}}\psi =\sum_{i=1}^{2}\rho(e_{i})\rho(\nu)\left(\hnabla_{e_{i}}\psi-\frac{1}{2}\rho(\mathcal{H}(e_{i},\cdot))\rho(\nu)\psi\right)\\\nonumber 
=\sum_{i=1}^{2}\rho(e_{i})\rho(\nu)\hnabla_{e_{i}}\psi-\frac{1}{2}\sum_{i=1}^{2}\rho(e_{i})c(\nu)\rho(\mathcal{H}(e_{i},\cdot))\rho(\nu)\psi
\end{eqnarray}
Now observe the following $(\mathcal{H}(e_{i},\cdot)\in T^{*} \widehat{\Sigma}$ and therefore $g(\mathcal{H}(e_{i},\cdot),\nu)=0$ and therefore 
\begin{eqnarray}
D\psi= \sum_{i=1}^{2}c(e_{i})c(\nu)\hnabla_{e_{i}}\psi+\frac{1}{2}\sum_{i=1}^{2}\rho(e_{i})\rho(\nu)\rho(\nu)\rho(\mathcal{H}(e_{i},\cdot))\psi\\\nonumber 
=\sum_{i=1}^{2}\rho(e_{i})\rho(\nu)\hnabla_{e_{i}}\psi+\frac{1}{2}H\psi=-\sum_{i=1}^{2}\rho(\nu)\rho(e_{i})\hnabla_{e_{i}}\psi+\frac{1}{2}H\psi\\\nonumber =-\rho(\nu)\sum_{i=1}^{3}\rho(e_{i})\hnabla_{e_{i}}\psi+\rho(\nu)\rho(\nu)\hnabla_{\nu}\psi+\frac{1}{2}H\psi\\\nonumber 
=-\rho(\nu)\widehat{D}\psi-\hnabla_{\nu}\psi+\frac{1}{2}H\psi,
\end{eqnarray}
where we have used the Clifford algebra relation \ref{eq:cliff} and the mean curvature $H=\tr{\mathcal{H}}$ and is defined with respect to the outward pointing normal to $ \widehat{\Sigma}$. 
\end{proof}

\noindent Now we derive the following identity 
\begin{proposition}
\label{main}
Let $\psi\in C^{\infty}(S^{\widehat{\Omega}})$ and $R^{\widehat{\Omega}}$ be the scalar curvature of $\widehat{\Omega}$ and $H$ be the mean curvature of the boundary $ \widehat{\Sigma}$ in $\widehat{\Omega}$.
 Then following Lichnerowicz type identity is verified by $\psi$ 
 \begin{eqnarray}
 \label{eq:lichnerowicz}
   \int_{ \widehat{\Sigma}}\left(\langle D\psi,\psi\rangle-\frac{1}{2}H|\psi|^{2}\right)\mu_{\widehat{\sigma}}= \frac{1}{4}\int_{\widehat{\Omega}}R^{\widehat{\Omega}}|\psi|^{2}\mu_{\widehat{\Omega}}-\int_{\widehat{\Omega}}|\widehat{D}\psi|^{2}\mu_{\widehat{\Omega}}+\int_{\widehat{\Omega}}|\hnabla\psi|^{2}\mu_{\widehat{\Omega}}  
 \end{eqnarray}
\end{proposition}
\begin{proof}
First, recall the Lichnerowicz identity on $\widehat{\Omega}$ 
\begin{eqnarray}
\label{eq:lich}
 \widehat{D}^{2}\psi=-\hnabla^{2}\psi+\frac{1}{4}R^{\widehat{\Omega}}\psi,   
\end{eqnarray}
for a spinor $\psi\in C^{\infty}(S^{\widehat{\Omega}})$. Let us consider an orthonormal frame $\{e_{i}\}_{i=1}^{3}$ on $\widehat{\Omega}$ with $e_{3}$ being the unit inside pointing normal $\nu$ of the boundary surface $\partial\widehat{\Omega}= \widehat{\Sigma}$. Now we evaluate the following using the compatibility property of the Hermitian inner product $\langle\cdot,\cdot\rangle$ on the spin bundle $S^{\widehat{\Omega}}$
\begin{eqnarray}
&&\hnabla_{i}\langle \rho(e_{i})\widehat{D}_{i}\psi,\psi\rangle+ \hnabla_{i}\langle\hnabla_{i}\psi,\psi\rangle\\&=&\langle \rho(e_{i})\hnabla_{i}\widehat{D}_{i}\psi,\psi\rangle\nonumber+\langle \rho(e_{i})\widehat{D}_{i}\psi,\hnabla_{i}\psi\rangle+\langle \hnabla^{2}\psi,\psi\rangle+\langle\hnabla_{i}\psi,\hnabla_{i}\psi\rangle\\\nonumber 
&=&\langle \widehat{D}^{2}\psi,\psi\rangle-\langle\widehat{D}_{i}\psi,\rho(e_{i})\hnabla_{i}\psi\rangle+\langle \hnabla^{2}\psi,\psi\rangle+\langle\hnabla_{i}\psi,\hnabla_{i}\psi\rangle\\\nonumber 
&=&\langle \widehat{D}^{2}\psi,\psi\rangle-\langle\widehat{D}_{i}\psi,\widehat{D}_{i}\psi\rangle+\langle \hnabla^{2}\psi,\psi\rangle+\langle\hnabla_{i}\psi,\hnabla_{i}\psi\rangle,
\end{eqnarray}
 where we have used the property \ref{eq:compat} and \ref{eq:clif} along with the fact that $\widehat{\nabla}_{i}(\rho(e_{i})\widehat{D}_{i}\psi)=\rho(e_{i})\hnabla_{i}\widehat{D}_{i}\psi$. Now we integrate the identity 
 \begin{eqnarray}
 \hnabla_{i}\langle \rho(e_{i})\widehat{D}_{i}\psi,\psi\rangle+ \hnabla_{i}\langle\hnabla_{i}\psi,\psi\rangle= \langle \widehat{D}^{2}\psi,\psi\rangle-\langle\widehat{D}_{i}\psi,\widehat{D}_{i}\psi\rangle+\langle \hnabla^{2}\psi,\psi\rangle+\langle\hnabla_{i}\psi,\hnabla_{i}\psi\rangle   
 \end{eqnarray}
 over $\widehat{\Omega}$ to obtain
 \begin{eqnarray}
 \int_{\widehat{\Omega}}\left(\hnabla_{i}\langle c(e_{i})\widehat{D}_{i}\psi,\psi\rangle+ \hnabla_{i}\langle\hnabla_{i}\psi,\psi\rangle\right)\mu_{\widehat{\Omega}}\\\nonumber 
 =\int_{\widehat{\Omega}}\left(\langle \widehat{D}^{2}\psi,\psi\rangle-\langle\widehat{D}_{i}\psi,\widehat{D}_{i}\psi\rangle+\langle \hnabla^{2}\psi,\psi\rangle+\langle\hnabla_{i}\psi,\hnabla_{i}\psi\rangle\right)\mu_{\widehat{\Omega}}
 \end{eqnarray}
 which by Stokes's theorem yields 
 \begin{eqnarray}
 -\int_{ \widehat{\Sigma}}\left(\langle c(\nu)\widehat{D}\psi,\psi\rangle+\langle\hnabla_{\nu}\psi,\psi\rangle\right)\mu_{\widehat{\sigma}}=\int_{\widehat{\Omega}}\left(\frac{1}{4}R^{\widehat{\Omega}}|\psi|^{2}-|\widehat{D}\psi|^{2}+|\hnabla\psi|^{2}\right)\mu_{\widehat{\Omega}}.
 \end{eqnarray}
 Using the proposition \ref{dirac1} relating the Dirac operator $D$ on $ \widehat{\Sigma}$ and the Dirac operator $\widehat{D}$ on $\widehat{\Omega}$, we obtain the desired identity
 \begin{eqnarray}
 \int_{ \widehat{\Sigma}}\left(\langle D\psi,\psi\rangle-\frac{1}{2}H|\psi|^{2}\right)\mu_{\widehat{\sigma}}=\int_{\widehat{\Omega}}\left(\frac{1}{4}R^{\widehat{\Omega}}|\psi|^{2}-|\widehat{D}\psi|^{2}+|\hnabla\psi|^{2}\right)\mu_{\widehat{\Omega}},
 \end{eqnarray}
 where we have used the Lichnerowicz identity \ref{eq:lich}. This concludes the proof of the spin identity we frequently use in this article.  
\end{proof}

\section{Dirac equation}
\label{spinorproof}
\noindent The Lichnerowicz type identity obtained in proposition \ref{main} is one of the main ingredients for the proof of the theorem \ref{1}. Notice that in the identity \ref{eq:lichnerowicz}, the right-hand side contains a non-positive definite term involving the action of the Dirac operator on the spinor $\psi$ in the bulk $\widehat{\Omega}$. To get rid of this term, we want to impose that $\widehat{D}\psi=0$ on $\widehat{\Omega}$ with a specified boundary condition on $ \widehat{\Sigma}:=\partial\widehat{\Omega}$. This issue is non-trivial since in dimensions greater than or equal to 2, the Dirac operator on a compact manifold with a boundary generally has an infinite-dimensional kernel. A simple example would be solving the Dirac equation on a disk in $\mathbb{C}$. This would essentially be finding the holomorphic and anti-holomorphic functions that form infinite dimensional vector space. More precisely, a Dirichlet boundary condition does not work for the Dirac equation on a compact manifold with a boundary. To fix this problem, one imposes a so-called elliptic boundary condition for the spinor on the boundary of the compact manifold $\widehat{\Omega}$. The elliptic boundary condition for the boundary value problem associated with the bulk Dirac operator $\widehat{D}$ is a pseudo-differential operator $P:L^{2}(S^{ \widehat{\Sigma}})\to L^{2}(V)$, where $V$ is a hermitian vector bundle on $ \widehat{\Sigma}$ such that the boundary value problem 
\begin{eqnarray}
    \widehat{D}\psi=\phi~\text{on}~\widehat{\Omega}\\
    P(\psi|_{ \widehat{\Sigma}})=\alpha~\text{on}~ \widehat{\Sigma}=\partial\widehat{\Omega}
\end{eqnarray}
has a unique solution given smooth data $\phi$ and $\alpha$ modulo finite dimensional kernel. In this section, we prove the Fredholm property of the Dirac operator $\widehat{D}$ with two boundary conditions: MIT Bag boundary condition (local) and APS boundary condition (non-local).  
\subsection{MIT Bag boundary condition}
\label{MITbc}
\noindent Let $S^{\widehat{\Omega}}$ denote the $C^{\infty}$ spin bundle over the connected oriented Riemannian $3-$ manifold $(\widehat{\Omega},g)$ with boundary $ \widehat{\Sigma}$. Let $S^{ \widehat{\Sigma}}$ denote the spin bundle over $ \widehat{\Sigma}$. Let $(e_{1},e_{2},e_{3}=\nu)$ be a $g-$orthonormal frae of $\widehat{\Omega}$. $\nu=e_{3}$ is the usual unit normal to $ \widehat{\Sigma}$ pointing inward. Define two point-wise projections 
\begin{eqnarray}
\label{eq:chiral}
 \Pi_{\pm}:S^{ \widehat{\Sigma}}\to S^{ \widehat{\Sigma}}\\
 \psi\mapsto \Pi_{\pm}(\psi):=\frac{1}{2}(\psi\pm \sqrt{-1}\rho(\nu)\psi).
\end{eqnarray}
Since the dimension of $ \widehat{\Sigma}$ is even, this is nothing but the projection onto the $\pm$ chirality subbundles $S^{ \widehat{\Sigma}}_{\pm}$ of $S^{ \widehat{\Sigma}}$. $\Pi_{+}$ and $\Pi_{-}$ are self-adjoint and orthogonal at each $S^{ \widehat{\Sigma}}_{x}$. \cite{Bar} provides a exposition of the technical details involving the MIT Bag boundary value problem of Dirac type operators. From now on we will write $\Pi_{+}:S^{ \widehat{\Sigma}}\to S^{ \widehat{\Sigma}}_{+}$ and $\Pi_{-}:S^{ \widehat{\Sigma}}\to S^{ \widehat{\Sigma}}_{-}$. Therefore we have for $\psi\in C^{\infty}(S^{ \widehat{\Sigma}})$ 
\begin{eqnarray}
 \psi=\Pi_{+}\psi+\Pi_{-}\psi.   
\end{eqnarray}
Dirac operator $D$ interchanges elements of $S^{ \widehat{\Sigma}}_{+}$ and $S^{ \widehat{\Sigma}}_{-}$ i.e., $D\psi_{+}\in S^{ \widehat{\Sigma}}_{-}$ and $D\psi_{-}\in S^{ \widehat{\Sigma}}_{+}$. This follows from the following fact
\begin{eqnarray}
 D(\gamma(\nu)\psi)=-\gamma(\nu)D\psi.   
\end{eqnarray}
The MIT bag boundary value problem is defined as follows 
\begin{eqnarray}
\widehat{D}\psi=0~on~\widehat{\Omega}\\
\Pi_{+}\psi=\Pi_{+}\alpha~on~ \widehat{\Sigma}
\end{eqnarray}
where $\alpha\in C^{\infty}(S^{ \widehat{\Sigma}})$. This boundary condition is symmetric with respect to $\Pi_{\pm}$ i.e., using $\Pi_{-}$ works equally well. The zero-order differential operator $\Pi_{+}$ provides elliptic boundary conditions for the Dirac operator. More precisely, the following lemma holds whose proof can be found in \cite{bar}.
\begin{lemma}
Let $\widehat{\Omega}$ be a connected oriented smooth Riemannian $3-$manifold with smooth boundary $ \widehat{\Sigma}$. Then the chirality operators $\Pi_{\pm}$ defined in \ref{eq:chiral} verify the Lopatinski-Shapiro condition i.e., they provide local elliptic boundary condition for the Dirac equation.  
\end{lemma}
\begin{proof}
See \cite{bar}.    
\end{proof}
Now we prove the Fredholm property of the Dirac operator with this boundary condition. 
\begin{proposition}
\label{isom}
Let $\widehat{\Omega}$ be a connected oriented  $3-$Riemannian manifold with connected smooth boundary $ \widehat{\Sigma}$ and $X$ be a smooth vector field on $\widehat{\Omega}$. Assume that the scalar curvature $R^{\widehat{\Omega}}$ of $\widehat{\Omega}$ verifies $ R^{\widehat{\Omega}}\geq 2|X|^{2}-2\text{div}X$. Then the inhomogeneous Dirac equation 
 \begin{eqnarray}
  \widehat{D}\psi=\Psi~\text{on}~\widehat{\Omega}~\\
  \Pi_{+}\psi=\Pi_{+}\alpha~\text{on}~ \widehat{\Sigma}
 \end{eqnarray}
 with $\Psi\in H^{s-1}(S^{\widehat{\Omega}})$ and $\alpha\in H^{s-\frac{1}{2}}(S^{ \widehat{\Sigma}}),s\geq 1$ has a unique solution $\psi\in H^{s}(S^{\widehat{\Omega}})\cap H^{s-\frac{1}{2}}(S^{ \widehat{\Sigma}})$.    
\end{proposition}
\begin{proof}
We prove that $\widehat{D}:H^{s}(S^{\widehat{\Omega}})\cap H^{s-\frac{1}{2}}(S^{ \widehat{\Sigma}})\to H^{s-1}(S^{\widehat{\Omega}})\cap H^{s-\frac{3}{2}}(S^{ \widehat{\Sigma}}),~s\geq 1$ is an isomorphism with this boundary condition. First, we show that the kernel 
\begin{eqnarray}
 \ker{\widehat{D}}:=\left\{\psi\in H^{s}(S^{\widehat{\Omega}})\cap H^{s-\frac{1}{2}}(S^{ \widehat{\Sigma}}) |\widehat{D}\psi=0~on ~\widehat{\Omega},~\Pi_{+}\psi=0~on~ \widehat{\Sigma}\right\}=\{0\}. 
\end{eqnarray}
This is straightforward to show using the integration by parts identity 
\begin{eqnarray}
\int_{\widehat{\Omega}}\langle \widehat{D}\psi,\sqrt{-1}\psi\rangle\mu_{g}=\int_{\widehat{\Omega}}\langle\psi,\sqrt{-1}\widehat{D}\psi\rangle\mu_{g}+\int_{ \widehat{\Sigma}}\langle\psi,\sqrt{-1}\rho(\nu)\psi\rangle\mu_{\widehat{\sigma}}.     
\end{eqnarray}
For $\psi\in \ker{\widehat{D}}$, we have 
\begin{eqnarray}
\int_{ \widehat{\Sigma}}\langle\psi,\sqrt{-1}\rho(\nu)\psi\rangle\mu_{\widehat{\sigma}}=0,~\psi=-\sqrt{-1}\rho(\nu)\psi~on~ \widehat{\Sigma}    
\end{eqnarray}
yielding 
\begin{eqnarray}
 \int_{ \widehat{\Sigma}}|\psi|^{2}\mu_{\widehat{\sigma}}=0\implies \psi=0.   
\end{eqnarray}
Therefore $\psi=0$ on $ \widehat{\Sigma}$. Now since $\widehat{D}\psi=0$ on $\widehat{\Omega}$, using Lichnerowicz formula $\widehat{D}^{2}=-\hnabla^{2}+\frac{1}{4}R^{\widehat{\Omega}}$
we obtain 
\begin{eqnarray}
\int_{\widehat{\Omega}}|\hnabla \psi|^{2}\mu_{g}+\frac{1}{4}R^{\widehat{\Omega}}|\psi|^{2}\mu_{g}=0 
\end{eqnarray}
due to $\psi=0$ on $ \widehat{\Sigma}$.
Now we add $\frac{1}{2}\int_{ \widehat{\Sigma}}\langle X,\nu\rangle|\psi|^{2}$. Since $\psi=0$ on $ \widehat{\Sigma}$, it does not alter anything. This yields through integration by parts 
\begin{eqnarray}
\int_{\widehat{\Omega}}\left(|\hnabla \psi|^{2}\mu_{\widehat{\Omega}}+\frac{1}{4}R^{\widehat{\Omega}}|\psi|^{2}+\frac{1}{2}\text{div}X|\psi|^{2}+\frac{1}{2}X(|\psi|^{2})\right)\mu_{g}=0  
\end{eqnarray}
or 
\begin{eqnarray}
\frac{1}{4} \int_{\widehat{\Omega}} (R^{\widehat{\Omega}}+2\text{div}X-2|X|^{2})|\psi|^{2}\mu_{g}+\frac{1}{2}\int_{\widehat{\Omega}}|\hnabla\psi|^{2}\mu_{g}\leq 0  
\end{eqnarray}
yielding $\psi=0$ on $\widehat{\Omega}$ since $R^{\widehat{\Omega}}+2\text{div}X-2|X|^{2}\geq 0$. This follows from the fact that $|\psi|^{2}$ is either identically zero or constant on $\widehat{\Omega}$ with vanishing boundary value. Therefore it must have zero length throughout $\widehat{\Omega}$ due to continuity. This proves the kernel $ \ker{\widehat{D}}=0$. Now we show that the co-kernel 
\begin{eqnarray}
\text{Co-ker}{\widehat{D}}:=\left\{\psi\in H^{s}(S^{\widehat{\Omega}})\cap H^{s-\frac{1}{2}}(S^{ \widehat{\Sigma}}) |\widehat{D}\psi=0~on ~\widehat{\Omega},~\Pi_{-}\psi=0~on~ \widehat{\Sigma}\right\}=\{0\}.    
\end{eqnarray}
But an exact same calculation as with the boundary condition $\Pi_{+}\psi=0$ yields 
$\psi=0$ on $ \widehat{\Sigma}$. Then the curvature condition together with the Lichrerowicz identity yields 
$\psi=0$ on $\widehat{\Omega}$ 
\begin{eqnarray}
 \text{Co-ker}\widehat{D}=0.   
\end{eqnarray}
\end{proof}

\begin{remark}
The MIT bag boundary condition is symmetric with respect to $\Pi_{+}$ or $\Pi_{-}$ data.    
\end{remark}

\begin{remark}
Note that the positivity condition on the mean curvature of $ \widehat{\Sigma}$ $H-\langle X,\nu\rangle>0$ is not required to prove the isomorphism property of the Dirac operator with MIT Bag boundary condition.     
\end{remark}
Proposition \ref{isom} allows one to utilize $\widehat{D}\psi=0$ in the Bochner identity \ref{main}. The following inequality follows as a consequence 
\begin{theorem}
\label{MIT1}
Let $(\widehat{\Omega},g)$ be a spacelike topological ball in the spacetime $(M,\widehat{g})$. Let $X$ be a vector field on $\widehat{\Omega}$ that verifies the following point-wise inequality 
\begin{eqnarray}
 R^{\widehat{\Omega}}\geq 2|X|^{2}-2\text{div}X,   
\end{eqnarray}
where $R^{\widehat{\Omega}}$ is the scalar curvature of $\widehat{\Omega}$. 
Let $\psi\in H^{s}(S^{\widehat{\Omega}})\cap H^{s-\frac{1}{2}}(S^{ \widehat{\Sigma}}),~s\geq 1$ solve the Dirac equation $\widehat{D}\psi=0$ in the bulk $\widehat{\Omega}$ with MIT bag boundary condition on $ \widehat{\Sigma}$ i.e., for $\alpha\in H^{s-\frac{1}{2}}(S^{ \widehat{\Sigma}})$
\begin{eqnarray}
 \widehat{D}\psi=0~on~\widehat{\Omega}\\
 \Pi_{+}(\psi)=\Pi_{+}(\alpha)~on~ \widehat{\Sigma}.
\end{eqnarray}
Then the following inequality is verified by $\psi$ 
\begin{eqnarray}
 \int_{ \widehat{\Sigma}}\langle D\psi,\psi\rangle \mu_{\widehat{\sigma}}\geq \frac{1}{2}\int_{ \widehat{\Sigma}}(H-\langle X,\nu\rangle)|\psi|^{2}\mu_{\widehat{\sigma}}.
\end{eqnarray}
The equality holds only if $X=0$, and $\psi$ on $ \widehat{\Sigma}$ is the restriction of a non-trivial parallel spinor in the interior $\widehat{\Omega}$. In such case 
\begin{eqnarray}
D\psi=\frac{1}{2}H\psi~\text{on}~ \widehat{\Sigma}    
\end{eqnarray}
\end{theorem}
\begin{proof}
The proof is a straightforward manipulation of the Bochner identity \ref{main}. First, use the proposition \ref{isom} to conclude that for $\alpha\in H^{s-\frac{1}{2}}(S^{ \widehat{\Sigma}}),s\geq 1$, the following boundary value problem has a unique solution
\begin{eqnarray}
 \widehat{D}\psi=0~on~\widehat{\Omega}\\
 \Pi_{+}(\psi)=\Pi_{+}(\alpha)~on~ \widehat{\Sigma}.   
\end{eqnarray}
Thus for such $\psi$, we have 
\begin{eqnarray}
\int_{ \widehat{\Sigma}}\left(\langle D\psi,\psi\rangle-\frac{1}{2}H|\psi|^{2}\right)\mu_{\widehat{\sigma}}= \frac{1}{4}\int_{\widehat{\Omega}}R^{\widehat{\Omega}}|\psi|^{2}\mu_{\widehat{\Omega}}+\int_{\widehat{\Omega}}|\hnabla\psi|^{2}\mu_{\widehat{\Omega}}.    
\end{eqnarray}
Now we perform the following trick. Addition of $\frac{1}{2}\int_{ \widehat{\Sigma}}\langle X,\nu\rangle|\psi|^{2}\mu_{\widehat{\sigma}}$ to the both side along with the use of the integration by parts identity 
\begin{eqnarray}
 \int_{ \widehat{\Sigma}}\langle X,\nu\rangle|\psi|^{2}=\int_{\widehat{\Omega}}\left(\text{div}X|\psi|^{2}+X(|\psi|^{2})\right)\mu_{\widehat{\Omega}}    
\end{eqnarray}
yields 
\begin{eqnarray}
\int_{ \widehat{\Sigma}}\left(\langle D\psi,\psi\rangle-\frac{1}{2}(H-\langle X,\nu\rangle)|\psi|^{2}\right)\mu_{\widehat{\sigma}}\\\nonumber=\frac{1}{4}\int_{\widehat{\Omega}}(R^{\widehat{\Omega}}+2\text{div}X)|\psi|^{2}\mu_{\widehat{\Omega}}+\int_{\widehat{\Omega}}|(|\hnabla\psi|^{2}+\frac{1}{2}X(|\psi|^{2}))\mu_{\widehat{\sigma}}\\\nonumber 
\geq \frac{1}{4}\int_{\widehat{\Omega}}(R^{\widehat{\Omega}}+2\text{div}X-2|X|^{2})|\psi|^{2}\mu_{\widehat{\Omega}}+\frac{1}{2}\int_{\widehat{\Omega}}|\hnabla\psi|^{2}\mu_{\widehat{\Omega}}\geq 0
\end{eqnarray}
due to the fact that $|\hnabla\psi|^{2}+\frac{1}{2}X(|\psi|^{2})=|\hnabla\psi|^{2}+\frac{1}{2}(\langle \hnabla_{X}\psi,\psi\rangle+\langle\psi,\hnabla_{X}\psi\rangle)\geq \frac{1}{2}|\hnabla\psi|^{2}-\frac{1}{2}|X|^{2}|\psi|^{2}$.
Therefore we have the desired inequality 
\begin{eqnarray}
 \int_{ \widehat{\Sigma}}\left(\langle D\psi,\psi\rangle-\frac{1}{2}(H-\langle X,\nu\rangle)|\psi|^{2}\right)\mu_{\widehat{\sigma}}\nonumber\geq \frac{1}{4}\int_{\widehat{\Omega}}(R^{\widehat{\Omega}}+2\text{div}X-2|X|^{2})|\psi|^{2}\mu_{\widehat{\Omega}}+\frac{1}{2}\int_{\widehat{\Omega}} |\widehat{\nabla}\psi|^{2}\mu_{\widehat{\Omega}}\\\nonumber 
\geq \frac{1}{4}\int_{\widehat{\Omega}}(R^{\widehat{\Omega}}+2\text{div}X-2|X|^{2})|\psi|^{2}\mu_{\widehat{\Omega}}\geq 0.
\end{eqnarray}
 Now equality holds only when 
\begin{eqnarray}
 \hnabla \psi=0,\\
 R^{\widehat{\Omega}}+2\text{div}X-2|X|^{2}=0.
\end{eqnarray}
Now $\hnabla\psi=0$ implies $\psi$ is a parallel Dirac spinor on $\widehat{\Omega}$. Therefore $\widehat{\Omega}$ must be Ricci flat and therefore flat since it has dimension $3$. Now recall the identity \ref{eq:hypersurface} from proposition \ref{dirac1}
\begin{eqnarray}
 D\psi=-\rho(\nu)\widehat{D}\psi-\widehat{\nabla}_{\nu}\psi+\frac{1}{2}H\psi    
\end{eqnarray}
yields 
\begin{eqnarray}
\label{eq:equalcase1}
D\psi =\frac{1}{2}H\psi~\text{on}~ \widehat{\Sigma}.
\end{eqnarray}
since $\widehat{D}\psi=0=\widehat{\nabla}_{\nu}\psi$. Therefore 
\begin{eqnarray}
 \int_{ \widehat{\Sigma}}\langle X,\nu\rangle|\psi|^{2} \mu_{\widehat{\sigma}}=0. 
\end{eqnarray}
which after integration by parts and parallelity condition $\hnabla\psi=0$ implies  
\begin{eqnarray}
\int_{\widehat{\Omega}}\text{div}X|\psi|^{2}\mu_{\widehat{\Omega}}=0    
\end{eqnarray}
Moreover, flatness implies $R^{\widehat{\Omega}}=0$ yielding 
\begin{eqnarray}
 \text{div}X-|X|^{2}=0   
\end{eqnarray}
and therefore 
\begin{eqnarray}
\label{eq:equalcase2}
\int_{\widehat{\Omega}}|X|^{2}|\psi|^{2}\mu_{\widehat{\Omega}}=0\implies X=0.  
\end{eqnarray}
\end{proof}
Theorem \ref{MIT1} has the following important corollary
\begin{corollary}
\label{MITC}
If $\psi$ is non-trivial and verifies the equality in the inequality \ref{MIT1} i.e., its restriction to the boundary $ \widehat{\Sigma}$ verifies $D\psi=\frac{1}{2}H\psi$, then the following is verified by its chiral components $\psi_{+}:=\Pi_{+}(\psi)$ and $\psi_{-}:=\Pi_{-}(\psi)$
\begin{eqnarray}
 \int_{ \widehat{\Sigma}}H|\psi_{+}|^{2}\mu_{\widehat{\sigma}}=\int_{ \widehat{\Sigma}}H|\psi_{-}|^{2}\mu_{\widehat{\sigma}}   
\end{eqnarray}
\end{corollary}
\begin{proof}
Recall the point-wise chiral decomposition of $\psi$
\begin{eqnarray}
 \psi=\Pi_{+}(\psi)+\Pi_{-}(\psi)=\psi_{+}+\psi_{-}   
\end{eqnarray}
that are point-wise orthogonal with respect to the Hermitian inner product $\langle\cdot,\cdot\rangle$
\begin{eqnarray}
 \langle \psi_{+},\psi_{-}\rangle=0.   
\end{eqnarray}
Also recall that the Dirac operator interchanges the subbundles $S^{ \widehat{\Sigma}}_{+}$ and $S^{ \widehat{\Sigma}}_{-}$. Now recall that on $ \widehat{\Sigma}$, $\psi$ verifies 
\begin{eqnarray}
\label{eq:equality1}
 D\psi=\frac{H}{2}\psi   
\end{eqnarray}
if the equality holds by virtue of the theorem \ref{MIT1}. Now apply $\Pi_{+},\Pi_{-}$ to both sides of \ref{eq:equality1} to yield 
\begin{eqnarray}
\label{eq:plus1}
 \Pi_{+}(D\psi)=\frac{H}{2}\Pi_{+}(\psi)~or~ D\Pi_{-}\psi=\frac{H}{2}\Pi_{+}(\psi)~or~D\psi_{-}=\frac{H}{2}\psi_{+}   
\end{eqnarray}
and 
\begin{eqnarray}
\label{eq:plus2}
 D\psi_{+}=\frac{H}{2}\psi_{-}.   
\end{eqnarray}
Now since $ \widehat{\Sigma}$ is closed, $D$ is self-adjoint and we have 
\begin{eqnarray}
 \int_{ \widehat{\Sigma}}\langle D\psi_{+},\psi_{-}\rangle\mu_{\widehat{\sigma}}=\int_{ \widehat{\Sigma}}\langle \psi_{+},D\psi_{-}\rangle\mu_{\widehat{\sigma}}   
\end{eqnarray}
Now use $\ref{eq:plus1}$ and $\ref{eq:plus2}$ to obtain 
\begin{eqnarray}
\int_{ \widehat{\Sigma}}\frac{H}{2}\langle \psi_{-},\psi_{-}\rangle\mu_{\widehat{\sigma}}=\int_{ \widehat{\Sigma}}\frac{H}{2}\langle \psi_{+},\psi_{+}\rangle\mu_{\widehat{\sigma}}.    
\end{eqnarray}
This concludes the proof.

\end{proof}

\subsection{APS boundary condition}
\noindent Here we work with the Atyiah-Patodi-Singer (APS) boundary condition \cite{atiyah1975spectral1,atiyah1975spectral2,atiyah1975spectral3} that was first utilized to prove the index theorem for compact manifolds with boundary. Here the choice of the pseudo-differential operator $P_{\lambda}$ is simply the $L^{2}$ orthogonal projection onto the subspace spanned by eigenvectors of the boundary Dirac operator $D$ whose eigenvalues are not smaller than $\lambda$ (note that $ \widehat{\Sigma}$ is closed and therefore the boundary operator $D$ is self-adjoint). \cite{atiyah1975spectral1,atiyah1975spectral2,atiyah1975spectral3} showed that the operator $P_{\lambda}$ is a zero order pseudo-differential operator. We invoke the following well-known proposition from 
\cite{atiyah1975spectral1,atiyah1975spectral2,atiyah1975spectral3,ginoux2009dirac} regarding the ellipticity of the APS boundary condition. 
\begin{proposition}\cite{atiyah1975spectral1,atiyah1975spectral2,atiyah1975spectral3,ginoux2009dirac}
Let $(\widehat{\Omega},\widetilde{g})$ be a compact connected oriented Riemannian $3-$ manifold with smooth boundary $ \widehat{\Sigma}:=\partial\widehat{\Omega}$. Then $APS$ boundary condition is elliptic and moreover, for $\lambda=0$, the spectrum of the Dirac operator is unbounded discrete.    
\end{proposition}
First, we note the following properties of the Sobolev spaces on $\widehat{\Omega}$ and $ \widehat{\Sigma}$. A well known fact is that all functions belonging to $H^{s}(\widehat{\Omega}),~s\in \mathbb{N}$ have traces on the smooth boundary $ \widehat{\Sigma}=\partial\widehat{\Omega}$ that belong to $H^{s-\frac{1}{2}}( \widehat{\Sigma})$. With a bit of technicality this passes to the case of spinors due to the compatibility of the spin connexion with the hermitian inner product $\langle\cdot,\cdot\rangle$ (in fact this generalizes to sections of smooth vector bundles where fiber metrics are compatible with the bundle connexions). More precisely the trace operator $\text{Tr}:H^{s}(\widehat{\Omega})\to H^{s-\frac{1}{2}}( \widehat{\Sigma})$ is bounded and subjective. 

On the boundary $ \widehat{\Sigma}$, $L^{2}(S^{ \widehat{\Sigma}})$ splits into two orthogonal subspaces. This is as follows. Let $D$ be the Dirac operator on $ \widehat{\Sigma}$. Consider the eigenvalue equation for $D$
\begin{eqnarray}
 D\psi=\lambda \psi,~\lambda\in \mathbb{R}.   
\end{eqnarray}
$ \widehat{\Sigma}$ is closed and $D$ is symmetric on $ \widehat{\Sigma}$. The spectrum $\{\lambda\in \mathbb{R}\}$ is symmetric with respect to zero. A generic section of $S^{ \widehat{\Sigma}}$ is written as $\sum_{i}a_{i}\psi_{i},~a_{i}\in \mathbb{C}$ and $D\psi_{i}=\lambda_{i}\psi_{i},\lambda_{i}\in \mathbb{R}$. One can split $L^{2}(S^{ \widehat{\Sigma}})$ into $L^{2}_{+}(S^{ \widehat{\Sigma}})$ and $L^{2}_{-}(S^{ \widehat{\Sigma}})$ as follows 
\begin{eqnarray}
L^{2}_{+}(S^{ \widehat{\Sigma}}):=\left\{\psi\in L^{2}(S^{ \widehat{\Sigma}})|\psi=\sum_{i=1}a_{i}\psi_{i},~D\psi_{i}=\lambda_{i}\psi_{i},\lambda_{i}\geq 0\right\},\\
L^{2}_{-}(S^{ \widehat{\Sigma}}):=\left\{\psi\in L^{2}(S^{ \widehat{\Sigma}})|\psi=\sum_{i=1}a_{i}\psi_{i},~D\psi_{i}=\lambda_{i}\psi_{i},\lambda_{i}< 0\right\}. 
\end{eqnarray}
Naturally $L^{2}_{+}(S^{ \widehat{\Sigma}})$ and $L^{2}_{-}(S^{ \widehat{\Sigma}})$ are $L^{2}$ orthogonal.
Let us denote by $P_{\geq 0},P_{<0}$ the projection operators
\begin{eqnarray}
 P_{\geq 0}: L^{2}(S^{ \widehat{\Sigma}})\to L^{2}_{+}(S^{ \widehat{\Sigma}}),\\
 P_{< 0}: L^{2}(S^{ \widehat{\Sigma}})\to L^{2}_{-}(S^{ \widehat{\Sigma}}).
\end{eqnarray}

\begin{theorem}
\label{dirac}
 Let $\widehat{\Omega}$ be a connected oriented  Riemannian $3-$ manifold with connected smooth boundary $ \widehat{\Sigma}$ and $X$ be a smooth vector field on $\widehat{\Omega}$. Assume that the scalar curvature $R^{\widehat{\Omega}}$ of $\widehat{\Omega}$ verifies $ R^{\widehat{\Omega}}\geq 2|X|^{2}-2\text{div}X$ and the mean curvature $H$ of $ \widehat{\Sigma}$ while viewed as an embedded surface in $\widehat{\Omega}$ verifies $ H-\langle X,\nu\rangle>0$. Then the inhomogeneous Dirac equation 
 \begin{eqnarray}
  \widehat{D}\psi=\Psi~\text{on}~\widehat{\Omega}~\\
  P_{\geq 0}\psi=P_{\geq 0}\alpha~\text{on}~ \widehat{\Sigma}
 \end{eqnarray}
 with $\Psi\in H^{s-1}(S^{\widehat{\Omega}})$ and $\alpha\in H^{s-\frac{1}{2}}(S^{ \widehat{\Sigma}}),s\geq 1$ has a unique solution $\psi\in H^{s}(S^{\widehat{\Omega}})\cap H^{s-\frac{1}{2}}(S^{ \widehat{\Sigma}})$. 
\end{theorem}

\begin{proof}
We want to first prove that the kernel 
\begin{eqnarray}
 \text{ker}\widehat{D}:=\left\{\psi\in H^{s}(S^{\widehat{\Omega}})\cap H^{s-\frac{1}{2}}(S^{ \widehat{\Sigma}}) |\widehat{D}\psi=0~on ~\widehat{\Omega},~P_{\geq 0}\psi=0~on~ \widehat{\Sigma}\right\}   
\end{eqnarray}
is trivial. We first show that if $\psi\in \text{ker}\widehat{D}$, then $\psi=0$ on the boundary $ \widehat{\Sigma}$. First we denote the restriction of $\psi$ on $ \widehat{\Sigma}$ by $\psi$ as well. The following $L^{2}$ orthogonal splitting holds on $ \widehat{\Sigma}$
\begin{eqnarray}
 \psi=P_{\geq 0}\psi+P_{<0}\psi. 
\end{eqnarray}
First, we want to show $P_{<0}\psi=0$ then that would imply $\psi=0$. Use the identity \ref{eq:lichnerowicz} with $\psi$ to yield 
\begin{eqnarray}
\int_{ \widehat{\Sigma}}\left(\langle D\psi,\psi\rangle-\frac{1}{2}H|\psi|^{2}\right)\mu_{\widehat{\sigma}}= \frac{1}{4}\int_{\widehat{\Omega}}R^{\widehat{\Omega}}|\psi|^{2}\mu_{\widehat{\Omega}}+\int_{\widehat{\Omega}}|\hnabla\psi|^{2}\mu_{\widehat{\Omega}}.    
\end{eqnarray}
Now addition of $\frac{1}{2}\int_{ \widehat{\Sigma}}\langle X,\nu\rangle|\psi|^{2}\mu_{\widehat{\sigma}}$ to the both side along with the use of the integration by parts identity 
\begin{eqnarray}
 \int_{ \widehat{\Sigma}}\langle X,\nu\rangle|\psi|^{2}=\int_{\widehat{\Omega}}\left(\text{div}X|\psi|^{2}+X(|\psi|^{2})\right)\mu_{\widehat{\Omega}}    
\end{eqnarray}
yields 
\begin{eqnarray}
\int_{ \widehat{\Sigma}}\left(\langle D\psi,\psi\rangle-\frac{1}{2}(H-\langle X,\nu\rangle)|\psi|^{2}\right)\mu_{\widehat{\sigma}}\nonumber=\frac{1}{4}\int_{\widehat{\Omega}}(R^{\widehat{\Omega}}+2\text{div}X)|\psi|^{2}\mu_{\widehat{\Omega}}+\int_{\widehat{\Omega}}|(|\hnabla\psi|^{2}+\frac{1}{2}X(|\psi|^{2}))\mu_{\widehat{\sigma}}\\\nonumber 
\geq \frac{1}{4}\int_{\widehat{\Omega}}(R^{\widehat{\Omega}}+2\text{div}X-2|X|^{2})|\psi|^{2}\mu_{\widehat{\Omega}}+\frac{1}{2}\int_{\widehat{\Omega}}|\hnabla\psi|^{2}\mu_{\widehat{\Omega}}\geq 0
\end{eqnarray}
due to the fact that $|\hnabla\psi|^{2}+\frac{1}{2}X(|\psi|^{2})=|\hnabla\psi|^{2}+\frac{1}{2}(\langle \hnabla_{X}\psi,\psi\rangle+\langle\psi,\hnabla_{X}\psi\rangle)\geq \frac{1}{2}|\hnabla\psi|^{2}-\frac{1}{2}|X|^{2}|\psi|^{2}$. Therefore we have the following 
\begin{eqnarray}
 \int_{ \widehat{\Sigma}}\left(\langle D\psi,\psi\rangle-\frac{1}{2}(H-\langle X,\nu\rangle)|\psi|^{2}\right)\mu_{\widehat{\sigma}}  \geq 0. 
\end{eqnarray}
Now we write $D\psi$ in the eigenspinor decomposition i.e., $D\psi=-\sum_{i}\lambda_{i}\psi_{i},~\lambda_{i}\in \mathbb{R}^{+}$ since APS boundary condition $P_{\geq 0}\psi=0$ removes the non-negative part of the spectra (note that here we are denoting the eigenvalues by $-\lambda_{i}$).
But on the boundary $ \widehat{\Sigma}$, $P_{\geq 0}\psi=0$ which yields 
\begin{eqnarray}
-\left(\int_{ \widehat{\Sigma}} \sum_{\lambda_{i}>0}\lambda_{i}|\psi_{i}|^{2}+\frac{1}{2}(H-\langle X,\nu\rangle)|P_{< 0}\psi|^{2}\right)\mu_{\widehat{\sigma}}\geq 0   
\end{eqnarray}
which yields $P_{< 0}\psi=0$ since $H-\langle X,\nu\rangle>0$. Therefore $\psi=0$ on $ \widehat{\Sigma}$. Following the exact procedure as in the proof of the theorem \ref{MIT1}, we obtain $\psi=0$ on $\widehat{\Omega}$. Note that once we obtain $\psi=0$ on the boundary, the process to prove $\psi=0$ on $\widehat{\Omega}$ is independent of the boundary condition used.

Now we prove the triviality of cokernel of $\widehat{D}$ in the context of APS boundary condition. The cokernel of $\widehat{D}$ is defined to be 
\begin{eqnarray}
    \text{coker}\widehat{D}:=\left\{\psi\in H^{s}(S^{\widehat{\Omega}})\cap H^{s-\frac{1}{2}}(S^{ \widehat{\Sigma}}) |\widehat{D}\psi=0~on ~\widehat{\Omega},~P_{> 0}\psi=0~on~ \widehat{\Sigma}\right\}
\end{eqnarray}
Now if $\psi\in  \text{coker}\widehat{D}$, then through the exact same procedure as the kernel case expect the loss of equality in $P_{\geq 0}\psi=0$, we obtain 
\begin{eqnarray}
-\left(\int_{ \widehat{\Sigma}} \sum_{\lambda_{i}\geq 0}\lambda_{i}|\psi_{i}|^{2}+\frac{1}{2}(H-\langle X,\nu\rangle)|\psi_{\leq 0}|^{2}\right)\mu_{\widehat{\sigma}}\geq 0    
\end{eqnarray}
which yields $\psi_{\leq 0}=0$ since $(H-\langle X,\nu\rangle)>0$. This proves the desired isomorphism property of the Dirac operator. 
\end{proof}

\noindent Now consider a smooth vector field $X$ on $\widehat{\Omega}$ that verifies $R^{\widehat{\Omega}}\geq 2|X|^{2}-2\text{div}X$. The theorem \ref{dirac} yields the following inequality  
\begin{corollary}
\label{spininequality}
Let $(\widehat{\Omega},g)$ be a spacelike topological ball in the spacetime $(M,\widehat{g})$. Let $X$ be a vector field on $\widehat{\Omega}$ that verifies the following point-wise inequality 
\begin{eqnarray}
 R^{\widehat{\Omega}}\geq 2|X|^{2}-2\text{div}X,   
\end{eqnarray}
where $R^{\widehat{\Omega}}$ is the scalar curvature of $\widehat{\Omega}$. Let the boundary of $\widehat{\Omega}$ be $ \widehat{\Sigma}$ with mean curvature $H$ while viewed as an embedded surface in $\widehat{\Omega}$ verifying
\begin{eqnarray}
 H-\langle X,\nu\rangle>0.   
\end{eqnarray}
Then for a spinor $\psi\in H^{s}(S^{\widehat{\Omega}}),~s\geq 1$ verifying the Dirac equation $\widehat{D}\psi=0$ in the bulk $\widehat{\Omega}$ and APS boundary condition on $ \widehat{\Sigma}$ as in theorem \ref{dirac}, we have the following inequality 
\begin{eqnarray}
\label{eq:main2}
 \int_{ \widehat{\Sigma}}\langle D\psi,\psi\rangle \mu_{\widehat{\sigma}}\geq \frac{1}{2}\int_{ \widehat{\Sigma}}(H-\langle X,\nu\rangle)|\psi|^{2}\mu_{\widehat{\sigma}}.
\end{eqnarray}
The equality holds only if $X=0$, and $\psi$ is a non-trivial parallel spinor in $\widehat{\Omega}$. In such case 
\begin{eqnarray}
D\psi=\frac{1}{2}H\psi~\text{on}~ \widehat{\Sigma}    
\end{eqnarray}
and $P_{\geq 0}\psi=P_{\geq 0}\alpha=\alpha$~on $ \widehat{\Sigma}$. Moreover, if the following condition is verified by $\psi$ solving the Dirac equation with APS boundary condition
\begin{eqnarray}
|P_{\geq 0}\psi|\leq |\psi|,   
\end{eqnarray}
then 
\begin{eqnarray}
 \int_{ \widehat{\Sigma}}\left(\langle DP_{\geq 0}\psi,P_{\geq 0}\psi\rangle-\frac{1}{2}(H-\langle X,\nu\rangle)|P_{\geq 0}\psi|^{2}\right) \mu_{\widehat{\sigma}}\geq 0
\end{eqnarray}
The equality holds only if $X=0$, and $P_{\geq 0}\psi=\psi$ (i.e., $P_{<0}\psi=0$) is the restriction on $ \widehat{\Sigma}$ of a non-trivial parallel spinor in the interior $\widehat{\Omega}$. In such case 
\begin{eqnarray}
DP_{\geq 0}\psi=\frac{1}{2}HP_{\geq 0}\psi~\text{on}~ \widehat{\Sigma}.    
\end{eqnarray}   
\end{corollary}
\begin{proof}
The proof relies on the proposition \ref{main} and theorem \ref{dirac}. First, the boundary value problem 
\begin{eqnarray}
\widehat{D}\psi=0~\text{on}~\widehat{\Omega}\\
P_{\geq 0}\psi=P_{\geq 0}\alpha~\text{on}~ \widehat{\Sigma}.
\end{eqnarray}
has a unique solution $\psi\in H^{s}(S^{\widehat{\Omega}})\cap H^{s-\frac{1}{2}}(S^{ \widehat{\Sigma}})$ given $\alpha\in H^{s-\frac{1}{2}}(S^{ \widehat{\Sigma}})$.  
Now we set $\widehat{D}\psi=0$ in the identity \ref{eq:lichnerowicz} and obtain 
\begin{eqnarray}
\label{eq:diracsolved}
   \int_{ \widehat{\Sigma}}\left(\langle D\psi,\psi\rangle-\frac{1}{2}H|\psi|^{2}\right)\mu_{\widehat{\sigma}}= \frac{1}{4}\int_{\widehat{\Omega}}R^{\widehat{\Omega}}|\psi|^{2}\mu_{\widehat{\Omega}}+\int_{\widehat{\Omega}}|\hnabla\psi|^{2}\mu_{\widehat{\Omega}}.  
\end{eqnarray}
From this expression proof of the inequality \ref{eq:main2} follows in an exactly similar way as the proof of \ref{MIT1}. Therefore, we avoid repeating the same calculations. For the second part,
consider the inequality 
\begin{eqnarray}
 \int_{ \widehat{\Sigma}}\langle D\psi,\psi\rangle \mu_{\widehat{\sigma}}\geq \frac{1}{2}\int_{ \widehat{\Sigma}}(H-\langle X,\nu\rangle)|\psi|^{2}\mu_{\widehat{\sigma}}  
\end{eqnarray}
and use $[D,P_{\geq 0}]=0=[D,P_{<0}]$ along with $\int_{ \widehat{\Sigma}}\langle DP_{\geq 0}\psi,P_{<0}\psi\rangle\mu_{\widehat{\sigma}}=0=\int_{ \widehat{\Sigma}}\langle DP_{< 0}\psi,P_{\geq 0}\psi\rangle\mu_{\widehat{\sigma}}$ to yield 
\begin{eqnarray}
\label{eq:hypo}
\int_{ \widehat{\Sigma}}\langle DP_{\geq 0}\psi,P_{\geq 0}\psi\rangle \mu_{\widehat{\sigma}}+\int_{ \widehat{\Sigma}}\langle DP_{< 0}\psi,P_{< 0}\psi\rangle \mu_{\widehat{\sigma}}\geq  
\frac{1}{2}\int_{ \widehat{\Sigma}}(H-\langle X,\nu\rangle)|\psi|^{2}\mu_{\widehat{\sigma}}.    
\end{eqnarray}
Now note that $P_{<0}\psi=\sum_{i}a_{i}\psi^{-}_{i}$, where $D\psi^{-}_{i}=\lambda_{i}\psi^{-}_{i},~\lambda_{i}<0$ and $\psi^{-}_{i}$ are $L^{2}$ normalized to have have norm $1$. Therefore 
\begin{eqnarray}
\label{eq:negative}
\int_{ \widehat{\Sigma}}\langle DP_{<0}\psi,P_{<0}\psi\rangle\mu_{\widehat{\sigma}}=\sum_{i}\lambda_{i}a^{2}_{i}<0.    
\end{eqnarray}
Thus 
\begin{eqnarray}
\int_{ \widehat{\Sigma}}\langle DP_{\geq 0}\psi,P_{\geq 0}\psi\rangle \geq \frac{1}{2}\int_{ \widehat{\Sigma}}(H-\langle X,\nu\rangle)|\psi|^{2}\mu_{\widehat{\sigma}}\geq \int_{ \widehat{\Sigma}}(H-\langle X,\nu\rangle|P_{\geq 0}\psi|^{2}\mu_{\widehat{\sigma}}    
\end{eqnarray}
where the last line follows from the hypothesis $|P_{\geq 0}\psi|\leq |\psi|$. Therefore, since $H-\langle X,\nu\rangle>0$
\begin{eqnarray}
\int_{ \widehat{\Sigma}}\langle DP_{\geq 0}\psi,P_{\geq 0}\psi\rangle \geq \frac{1}{2}\int_{ \widehat{\Sigma}}(H-\langle X,\nu\rangle|P_{\geq 0}\psi|^{2}\mu_{\widehat{\sigma}}.     
\end{eqnarray}
For the equality case, consider the inequality \ref{eq:hypo}
\begin{eqnarray}
\int_{ \widehat{\Sigma}}\langle DP_{\geq 0}\psi,P_{\geq 0}\psi\rangle \mu_{\widehat{\sigma}}+\int_{ \widehat{\Sigma}}\langle DP_{< 0}\psi,P_{< 0}\psi\rangle \mu_{\widehat{\sigma}}\geq  
\frac{1}{2}\int_{ \widehat{\Sigma}}(H-\langle X,\nu\rangle)|\psi|^{2}\mu_{\widehat{\sigma}}\geq\\\nonumber  \frac{1}{2}\int_{ \widehat{\Sigma}}(H-\langle X,\nu\rangle|P_{\geq 0}\psi|^{2}\mu_{\widehat{\sigma}}.     
\end{eqnarray}
or 
using \ref{eq:negative} 
\begin{eqnarray}
\label{eq:hypo1}
\int_{ \widehat{\Sigma}}\langle DP_{\geq 0}\psi,P_{\geq 0}\psi\rangle \mu_{\widehat{\sigma}} +\sum_{i}\lambda_{i}a^{2}_{i}\geq   \frac{1}{2}\int_{ \widehat{\Sigma}}(H-\langle X,\nu\rangle|P_{\geq 0}\psi|^{2}\mu_{\widehat{\sigma}}.  
\end{eqnarray}
Now when 
\begin{eqnarray}
\int_{ \widehat{\Sigma}}\langle DP_{\geq 0}\psi,P_{\geq 0}\psi\rangle = \int_{ \widehat{\Sigma}}(H-\langle X,\nu\rangle|P_{\geq 0}\psi|^{2}\mu_{\widehat{\sigma}},
\end{eqnarray}
\ref{eq:hypo1} implies
\begin{eqnarray}
\sum_{i}\lambda_{i}a^{2}_{i}\geq 0    
\end{eqnarray}
which is impossible unless $a_{i}=0~\forall i$ by \ref{eq:negative}. Therefore $P_{<0}\psi=0$. Moreover, the exact similar analysis with the Bochner formula yields 
\begin{eqnarray}
 \int_{ \widehat{\Sigma}}\left(\langle DP_{\geq 0}\psi,P_{\geq 0}\psi\rangle-\frac{1}{2}(H-\langle X,\nu\rangle)|P_{\geq 0}\psi|^{2}\right)\mu_{\widehat{\sigma}}\\\nonumber\geq \frac{1}{4}\int_{\widehat{\Omega}}(R^{\widehat{\Omega}}+2\text{div}X-2|X|^{2})|\psi|^{2}\mu_{\widehat{\Omega}}+\frac{1}{2}\int_{\widehat{\Omega}} |\widehat{\nabla}\psi|^{2}\mu_{\widehat{\Omega}}\\\nonumber 
\geq \frac{1}{4}\int_{\widehat{\Omega}}(R^{\widehat{\Omega}}+2\text{div}X-2|X|^{2})|\psi|^{2}\mu_{\widehat{\Omega}}\geq 0   
\end{eqnarray}
and therefore the equality case implies $\psi$ is parallel in $\widehat{\Omega}$ and $\widehat{\Omega}$ is Ricci flat (this is exactly similar to that of the proof of equality case in theorem 1\ref{MIT1}). Application of the proposition \ref{dirac1} to $\psi$ with $P_{<0}\psi=0$ on the boundary $ \widehat{\Sigma}$ implies 
\begin{eqnarray}
DP_{\geq 0}\psi=\frac{1}{2}HP_{\geq 0}\psi~\text{on}~ \widehat{\Sigma}.   
\end{eqnarray}
\end{proof}

\section{Quasi-local mass and its positivity}
We describe the notion of quasi-local mass described by Wang and Yau \cite{yau}.
 The details can be found in the original article of Wang $\&$ Yau \cite{yau}. Here we state the necessary ideas that lead to the construction of the Wang-Yau quasi-local mass. The aim is to define the gravitational energy contained in a spacelike region bounded by a topological $2-$sphere in a $3+1$ dimensional physical spacetime $(M,\widehat{g})$ verifying dominant energy condition. We assume that the spacetime $M$ is globally hyperbolic, that is, it can be expressed as the product $\widetilde{M}\times \mathbb{R}$ where $\widetilde{M}$ is a Cauchy hypersurface. We write the spacetime metric $\widehat{g}$ in the following form 
\begin{eqnarray}
\widehat{g}:=-\beta^{2}dt\otimes dt+g_{ij}(dx^{i}+\gamma^{i}dt)\otimes (dx^{j}+\gamma^{j}dt),   
\end{eqnarray}
in a local chart $\{t,x^{i}\}_{i=1}^{3}$. Here $\beta$ and $\gamma:=\gamma^{i}\frac{\partial}{\partial x^{i}}$
are the lapse function and the shift vector field, respectively. Here we will be interested in Einstein's constraint equations primarily. Note that the Gauss and Codazzi equations for the spacetime $M$ foliated by $\widetilde{M}$ yield the constraint equations on each $t=\text{constant}$ slice
\begin{eqnarray}
 R(g)-P_{ij}P^{ij}+(\text{tr}_{g}P)^{2}=2\mu\\
 \widehat{\nabla}^{j}P_{ij}-\nabla_{i}\text{tr}_{g}P=J_{i},
\end{eqnarray}
where $\mu$ and $J:=J_{i}dx^{i}$ are the local energy and momentum density of the matter (or radiation) sources present in the spacetime. The dominant energy condition states $\mu\geq \sqrt{g(J,J)}$. $P_{ij}$ is the extrinsic curvature or the second fundamental form of $\widetilde{M}$ in $M=\widetilde{M}\times \mathbb{R}$. Once we are equipped with this information about the spacetime $(M,\widehat{g})$ (which is assumed to be Einsteinian spacetime i.e., the spacetime metric verifies Einstein's constraint and evolution equations), we can define the energy contained within a membrane $ \widehat{\Sigma}\subset M$ i.e., a spacelike topological $2-$sphere. The Wang-Yau quasi-local mass precisely quantifies this energy through a Hamilton-Jacobi analysis (see \cite{yau,yau1} for the detail on how such energy expression is obtained). 
\subsection{Definition of Wang-Yau Quasi-local Energy}
\label{reduction}
Wang-Yau quasi-local energy is the most consistent notion of quasi-local energy to date. It is defined for a topological $2-$sphere $\Sigma$ bounding a spacelike topological ball $\Omega$ in the spacetime $(M,\widehat{g})$. Energy content in gravity is essentially a measure of the non-triviality of the Riemann curvature. Therefore a comparison with a background spacetime (e.g., Minkowski spacetime) is required to define any sensible notion of energy. The main idea behind defining a quasi-local energy associated with a membrane $\Sigma$ is to compare the extrinsic geometries of the membrane $\Sigma$ in the physical spacetime $(M,\widehat{g})$ and a reference spacetime (e.g., Minkowski spacetime in the current context). Here we provide a concise description of the Wang-Yau quasi-local energy. For exact detail, one is referred to \cite{yau}.
First, we recall the following definition of a generalized mean curvature as given in \cite{yau}. 
\begin{definition}[Wang-Yau] \cite{yau}
\label{def1} Let $(M,\widehat{g})$ be a connected time oriented $C^{\infty}$ Lorentzian manifold verifying dominant energy condition \ref{eq:dominant}. Suppose $i: \Sigma\hookrightarrow M$ is an embedded space-like two-surface with induced metric $\sigma$. Given a
smooth function $\tau$ on $\Sigma$ and a space-like normal $e_{3}$, the generalized mean curvature $h(\Sigma,i,\tau,e_{3})$ associated with these data is defined to be
\begin{eqnarray}
 h(\Sigma,i,\tau,e_{3}):=-\sqrt{1+|\nabla\tau|^{2}}\langle \mathbf{H},e_{3}\rangle-\alpha_{e_{3}}(\nabla\tau),   
\end{eqnarray}
where $\mathbf{H}$ is the mean curvature vector of $\Sigma$ in $(M,\widehat{g})$ and and $\alpha_{e_{3}}$ is the connexion 1-form (see
Definition 1.1) of the normal bundle of $\Sigma$ in $M$ determined by $e_{3}$ and the future-directed time-like unit normal $e_{4}$ orthogonal to $e_{3}$.
\end{definition}

\noindent Now we consider the isometric embedding of $( \Sigma,\sigma)$ into the Minkowski space $(\mathbb{R}^{1,3},\eta)$. Let us denote the standard Lorentzian inner product on $\mathbb{R}^{1,3}$ by $\langle\cdot,\cdot\rangle$ and  let $T_{0}$ be a constant timeline unit vector in $\mathbb{R}^{1,3}$. First, we need to address the existence of such an embedding.

\begin{theorem}[Wang-Yau] \cite{yau}
\label{embedding}
    Let $\chi$ be a function on $(\Sigma, \sigma)$ with $\int_{\Sigma}\chi\mu_{\sigma}=0$. Let $\tau$ be a potential function of $\chi$ i.e., $\Delta_{ \widehat{\Sigma}}\tau=\chi$. Suppose the Gaussian curvature $K$ of $(\Sigma,\sigma)$ verifies the inequality $K+(1+|\nabla\tau|^{2}_{ \sigma})^{-1}\det(\nabla^{2}\tau)>0$,    
    then there exists a space-like isometric embedding $i_{0}: \Sigma\to \mathbb{R}^{1,3}$ such that the mean curvature $\textbf{H}_{0}$ of the embedded surface $(i_{0}(\Sigma), \sigma)$ verifies $\langle \textbf{H}_{0},T_{0}\rangle=-\chi$.
\end{theorem}
The main idea behind the proof of this theorem is to consider the projection of $\widehat{i}_{0}: \Sigma\to \mathbb{R}^{3}$ onto the orthogonal complement of $T_{0}$ i.e., $\widehat{i}_{0}:=i_{0}-\tau T_{0}$. Let us denote the projected surface by $ \widehat{\Sigma}$. The induced metric on $ \widehat{\Sigma}$ is given by $\sigma:=\sigma+d\tau\otimes d\tau$ and the Gauss curvature $\widehat{K}$ as $(1+|\nabla\tau|^{2}_{\sigma})^{-1}(K+(1+|\nabla\tau|^{2}_{ \sigma})^{-1}\det(\nabla^{2}\tau))$ which is positive by the assumption of the theorem \ref{embedding}. Therefore, the existence of $\widehat{i}_{0}$ follows by the classic theorem of Nirenberg and Pogolerov \cite{pogorelov1952regularity}. Now existence of $i_{0}$ can be established as follows. One starts with a metric $ \sigma$ and a function $\chi$ and solves for $\Delta_{ \sigma}\tau=\lambda$ (unique solution exists since $ \Sigma$ is compact, the kernel of $\Delta_{ \sigma}$ is trivial modulo constants). Now one constructs the new metric $ \widehat{\sigma}= \sigma+d\tau\otimes d\tau$ and finds that its Gauss curvature $\widehat{K}$ is again $(1+|\nabla\tau|^{2}_{ \sigma})^{-1}(K+(1+|\nabla\tau|^{2}_{ \sigma})^{-1}\det(\nabla^{2}\tau))$ which is strictly positive by the hypothesis. Therefore the embedding into $\mathbb{R}^{3}$ is established. The desired embedding $i_{0}$ is found as $i_{0}=\widehat{i}_{0}+\tau T_{0}$. 

\noindent Now the generalized mean curvature of the surface $(i_{0}( \Sigma), \sigma)$ in the Minkowski space is defined in a similar fashion
\begin{eqnarray}
h( \Sigma,i_{0},\tau,(e_{3})_{0}):=  -\sqrt{1+|\nabla\tau|^{2}}\langle \mathbf{H}_{0},(e_{3})_{0}\rangle-\widehat{\alpha}_{(e_{3})_{0}}(\nabla\tau)  
\end{eqnarray}
where $\widehat{\alpha}$ is connection of the normal bundle of $i_{0}(\Sigma)$ induced by the flat connection $\nabla^{\mathbb{R}^{1,3}}$ of the Minkowski space,  $(e_{3})_{0}$ is the spacelike unit normal vector to $(i_{0}(\Sigma), \sigma)$, and $\mathbf{H}_{0}$ is the mean curvature of $i_{0}( \Sigma)$ as an embedded spacelike surface in the Minkowski space.

The Wang-Yau quasi-local energy is defined to be the following 
\begin{eqnarray}
8\pi M^{WY}:=\inf_{\tau}\int_{ \widehat{\Sigma}}\left(h( \widehat{\Sigma},i_{0},\tau,(e_{3})_{0})-h( \widehat{\Sigma},i,\tau,e_{3})\right)\mu_{\widehat{\sigma}}.   
\end{eqnarray}

\noindent Strictly speaking, in the definition of $h( \widehat{\Sigma},i,\tau,e_{3})$, there is still redundancy left due to the boost transformation of the normal bundle of $ \widehat{\Sigma}$ in the physical spacetime $(M,\widehat{g})$. In fact, \cite{yau} minimizes $h( \widehat{\Sigma},i,\tau,e_{3})$ over the boost-angle ($h( \widehat{\Sigma},i,\tau,e_{3})$ is a convex function of the boost angle as described in the introduction). Here, we aim to prove the positivity of the entity $h( \widehat{\Sigma},i_{0},\tau,(e_{3})_{0})-h( \widehat{\Sigma},i,\tau,e_{3})$ for any $\tau$ verifying admissibility condition \ref{admi}. From now on we define the new entity the reduced Wang-Yau energy 
\begin{eqnarray}
 8\pi M^{WY}_{\tau}:=\int_{ \widehat{\Sigma}}\left(h( \widehat{\Sigma},i_{0},\tau,(e_{3})_{0})-h( \widehat{\Sigma},i,\tau,e_{3})\right)\mu_{\widehat{\sigma}}\\
 =\int_{ \widehat{\Sigma}}\left(-\sqrt{1+|\nabla\tau|^{2}}\langle \mathbf{H}_{0},(e_{3})_{0}\rangle-\widehat{\alpha}_{(e_{3})_{0}}(\nabla\tau) \right)\mu_{\widehat{\sigma}}\\\nonumber -\int_{ \widehat{\Sigma}}\left(-\sqrt{1+|\nabla\tau|^{2}}\langle \mathbf{H},e_{3}\rangle-\alpha_{e_{3}}(\nabla\tau)\right)\mu_{\widehat{\sigma}}.
\end{eqnarray}
If we can prove the non-negativity of $M^{WY}_{\tau}$ for every admissible $\tau$ (as in definition \ref{admi}), then that would imply $M^{WY}\geq 0$.

Our goal is to apply the spin inequalities where essentially a comparison is made between the mean curvature of the surface $ \widehat{\Sigma}$ embedded in the physical spacetime and with its isometrically embedded image in the flat space $\mathbb{R}^{3}$. However, the expression of the Wang-Yau energy is not of this particular form. Therefore, an application of the inequalities proven in theorem \ref{MIT1} in the context of MIT Bag boundary condition or in corollary \ref{spininequality} in the context of non-local APS boundary condition does not immediately help. Therefore, we need to reduce the Wang-Yau quasi-local energy to a form where such an application becomes possible.

The most important observation to this end is that the Minkowski generalized mean curvature $\int_{ \widehat{\Sigma}}h( \widehat{\Sigma},i_{0},\tau,(e_{3})_{0})\mu_{\widehat{\sigma}}$ can be re-written as the integral of the mean curvature of the projection $ \widehat{\Sigma}$ of $ \widehat{\Sigma}$ onto the compliment of $T_{0}$ i.e., onto $\mathbb{R}^{3}$. Denote the spacelike topological ball that $ \widehat{\Sigma}$ bounds by $\widehat{\Omega}$. First, consider an embedding $i_{0}: \widehat{\Sigma}\hookrightarrow \mathbb{R}^{1,3}$ is an embedding and $\tau:=\langle i_{0},T_{0}\rangle$ is the restriction of the time function associated with the time vector field $T_{0}$. The outward unit normal vector $\widehat{\nu}$ of $ \widehat{\Sigma}$ in $\mathbb{R}^{3}$ and $T_{0}$ form a basis of the normal bundle of $ \widehat{\Sigma}$ in $\mathbb{R}^{1,3}$. One can parallel transport $\widehat{\nu}$ along the integral curves of $T_{0}$ to make it a vector field in $\mathbb{R}^{1,3}$ and denote this vector field as $\widehat{e}_{3}$. Let the mean curvature of $ \widehat{\Sigma}$ in $\mathbb{R}^{3}$ be $\widehat{k}_{0}$. Then the following proposition holds
\begin{proposition}
\label{mean1}
Let $i_{0}: \widehat{\Sigma}\hookrightarrow \mathbb{R}^{1,3}$ be the isometric embedding of $( \widehat{\Sigma}, \widehat{\Sigma})$ into the Minkowski space. Let $T_{0}$ be the canonical unit time-like vector field in $\mathbb{R}^{1,3}$. Consider the projection $ \widehat{\Sigma}$ of $i_{0}( \widehat{\Sigma})$ onto the complement of $T_{0}$ i.e., onto $T_{0}$ orthogonal hypersurface $\mathbb{R}^{3}$. Let $\tau:=\langle i_{0},T_{0}\rangle$ is the restriction of the time function associated with the vector field $T_{0}$ onto $ \widehat{\Sigma}$ and assume that it verifies the admissibility criteria \ref{admi}. Let the mean curvature of $ \widehat{\Sigma}$ in $\mathbb{R}^{3}$ be $\widehat{k}_{0}$. Then, $\widehat{k}_{0}$ verifies the following integrated identity  
 \begin{eqnarray}
\int_{ \widehat{\Sigma}}\widehat{k}_{0}\mu_{\widehat{\sigma}}= \int_{ \widehat{\Sigma}}\left(-\langle k_{0},\widehat{e}_{3}\rangle\sqrt{1+|\nabla\tau|^{2}}-\alpha_{\widehat{e}_{3}} (\nabla\tau) \right)\mu_{\widehat{\sigma}}.  
 \end{eqnarray}
 Here the induced metric on $ \widehat{\Sigma}$ is $ \widehat{\Sigma}+d\tau\otimes d\tau$ and $ \widehat{\Sigma}$ is the metric on $ \widehat{\Sigma}$.
\end{proposition}
\begin{proof}
 First, consider the flat Lorentzian metric compatible connection $\nabla^{1,3}$ of the Minkowski space $\mathbb{R}^{1,3}$. Let $ \widehat{\Sigma}$ be the projection of $ \Sigma$ onto the $T_{0}$ complement i.e., a constant time Cauchy hypersurface in $R^{1,3}$. Let $\widehat{e}_{I},I=1,2$ be an orthonormal basis of the tangent bundle of $ \widehat{\Sigma}$. Let $\widehat{\nu}$ be the spacelike unit normal vector to $ \widehat{\Sigma}$. Extend $\widehat{\nu}$ along $T_{0}$ by parallel translation and denote it by $\widehat{e}_{3}$. The mean curvature of $\widehat{\Sigma}$ in the complement of $T_{0}$ is computed to be 
 \begin{eqnarray}
  k_{0}=\langle\nabla^{1,3}_{\widehat{e}_{I}}\widehat{\nu},\widehat{e}_{I}\rangle= \langle\nabla^{1,3}_{\widehat{e}_{I}}\widehat{\nu},\widehat{e}_{I}\rangle+\langle\nabla^{1,3}_{\widehat{\nu}}\widehat{\nu},\widehat{\nu}\rangle-\langle\nabla^{1,3}_{T_{0}}\widehat{\nu},T_{0}\rangle  
 \end{eqnarray}
 since the last two terms are identically zero. Therefore in the local coordinate symbol, we have 
 \begin{eqnarray}
  k_{0}=g^{\alpha\beta}\langle \nabla^{1,3}_{e_{\alpha}}\widehat{\nu},e_{\beta}\rangle    
 \end{eqnarray}
 for any orthonormal frame $e_{\alpha}$ of $\mathbb{R}^{1,3}$, where $g^{\alpha\beta}$ is the inverse of the metric $g_{\alpha\beta}=\langle e_{\alpha},e_{\beta}\rangle$. Now $\widehat{e}_{3}=\widehat{\nu}$ may be considered as a spacelike normal vector field along $ \widehat{\Sigma}$. Pick an orthonormal basis $\{e_{1},e_{2}\}$ tangent to $ \widehat{\Sigma}$. Let $\widehat{e}_{4}=\frac{1}{\sqrt{1+|\nabla\tau|^{2}}}(T_{0}-T^{\perp}_{0})$ be the future directed unit normal vector in the direction of the normal part of $T_{0}$. By definition 
 \begin{eqnarray}
  T^{\perp}_{0}=-\nabla\tau\cdot\{\widehat{e}_{3},\widehat{e}_{4}\}   
 \end{eqnarray}
 form an orthonormal basis of the normal bundle of $ \widehat{\Sigma}$. One derives 
 \begin{eqnarray}
 \label{eq:final1}
  k_{0}=\langle\nabla^{1,3}_{e_{I}}\widehat{e}_{3},\widehat{e}_{I}\rangle-\langle\nabla^{1,3}_{\widehat{e}_{4}}\widehat{e}_{3},\widehat{e}_{4}\rangle=\langle \mathbf{H}_{0},\widehat{e}_{3}\rangle-\frac{1}{\sqrt{1+|\nabla\tau|^{2}}}\langle\nabla^{1,3}_{\nabla\tau}\widehat{e}_{3},\widehat{e}_{4}\rangle \\\nonumber 
  =\langle \mathbf{H}_{0},\widehat{e}_{3}\rangle-\frac{1}{\sqrt{1+|\nabla\tau|^{2}}}\alpha_{\widehat{e}_{3}} (\nabla\tau) 
 \end{eqnarray}
 since $\widehat{\nu}$ is extended along $T_{0}$ by parallel translation. Now recall the volume form of $ \sigma$ and $ \widehat{\sigma}$ are related by
 \begin{eqnarray}
  \mu_{\sigma}=\frac{1}{\sqrt{1+|\nabla\tau|^{2}}}\mu_{\widehat{\sigma}}   
 \end{eqnarray}
Therefore integrating \ref{eq:final1} over $\Sigma$, one obtains 
\begin{eqnarray}
\int_{\Sigma}\left(-\langle \mathbf{H}_{0},\widehat{e}_{3}\rangle\sqrt{1+|\nabla\tau|^{2}}-\alpha_{\widehat{e}_{3}} (\nabla\tau) \right)\mu_{\sigma}=\int_{ \widehat{\Sigma}}\widehat{k}_{0}\mu_{\widehat{\sigma}}.     
\end{eqnarray}
  
\end{proof}

\noindent This only takes care of one term in the Wang-Yau mass i.e., the Minkowski contribution is expressible in terms of the mean curvature of a projected surface onto $\mathbb{R}^{3}$. We still need to take care of the physical contribution term i.e., the generalized mean curvature of $ \Sigma$ while viewed as an embedded surface in the spacetime $(M,\widehat{g})$.

Now one hopes to find a similar expression for the physical spacetime entity $\int_{\Sigma}h(\Sigma,i,\tau,e_{3})\mu_{\sigma}$. But in the physical spacetime, one needs to account for the \text{momentum} information (i.e., the second fundamental form of the hypersurface $\Omega$ in $(M,\widehat{g})$). \cite{yau} derives an almost similar identity by projecting $\Sigma$ onto a constant time hypersurface with extra terms accounting for the momentum information. As we shall see, the appearance of this extra momentum contribution would necessitate invoking the dominant energy condition. Consider the initial data set $(\Omega,g_{ij},P_{ij})$, where $P_{ij}$ is the second fundamental form of $\Omega$ in the physical spacetime $(M,\widehat{g})$ and plays the role of momentum for the metric $g_{ij}$. Now consider the product $\Omega\times \mathbb{R}$ and extend $P_{ij}$ by parallel translation along the $\widehat{\Omega}$-normal field $e_{4}$. We denote this entity by $P_{ij}$ as well. Now we seek a hypersurface $\widehat{\Omega}$ in $\Omega\times \mathbb{R}$ defined as the graph of a function
$f$ over $\Omega$, such that the mean curvature of $\widehat{\Omega}$
in $\Omega\times \mathbb{R}$ is the same as the trace of the restriction of $P$ to $\widehat{\Omega}$. This is accomplished in \cite{yau} by solving Jang's equation verified by $f$ (Jang's equation enjoys a rich history: Jang’s equation was proposed by Jang \cite{jang1978positivity} in an attempt to solve
the positive energy conjecture. Schoen and Yau came up with different geometric interpretations, studied the equation in full, and applied them to their proof of the positive
mass theorem \cite{schoen1979proof,schoen1981proof}). First, pick an orthonormal basis $\{\widetilde{e}_{\alpha}\}_{\alpha=1}^{4}$ for the tangent space of $\Omega\times \mathbb{R}$ along $\widehat{\Omega}$ such that $\{\widetilde{e}_{i}\}_{i=1}^{3}$ is tangent to $\widehat{\Omega}$ and $\widetilde{e}_{4}$ is the downward unit normal to $\widehat{\Omega}$. If the induced metric on $\Omega$ is $g_{ij}$, then the induced metric on $\widehat{\Omega}$ is $\widetilde{g}_{ij}=g_{ij}+\partial_{i}f\partial_{j}f$. Jang's equations states
\begin{eqnarray}
\label{eq:jang}
  \sum_{i=1}^{3}\langle\widetilde{\nabla}_{\widetilde{e}_{i}}\widetilde{e}_{4},\widetilde{e}_{i}\rangle=\sum_{i=1}^{3}P(\widetilde{e}_{i},\widetilde{e}_{i})  
\end{eqnarray}
or more explicitly 
\begin{eqnarray}
\sum_{i,j=1}^{3}\left(g^{ij}-\frac{f^{i}f^{j}}{1+|\hnabla f|^{2}}\right)\left(\frac{\hnabla_{i}\partial_{j}f}{(1+|\hnabla f|^{2})^{\frac{1}{2}}}-P_{ij}\right)=0    
\end{eqnarray}
where $\widetilde{\nabla}$ is the usual Levi-Civita connexion of the product spacetime $\Omega\times \mathbb{R}$. Now let $\tau$ be a smooth function on $\Sigma$. We consider a solution $f$ to Jangs equation \ref{eq:jang} in $\Omega$ that verifies the Dirichlet boundary condition $f=\tau$ on $ \Sigma$. Let $\widehat{\Sigma}$ be the graph of $\tau$ over $ \Sigma$ and $\widehat{\Omega}$ be the graph of $f$ over $\Omega$ such that $\partial\widehat{\Omega}=\widehat{\Sigma}$. The metric on $\widehat{\Sigma}$ is $ \sigma+d\tau\otimes d\tau$. Jang's equation is a quasi-linear elliptic PDE for $f$ and typically its solution blows up in the presence of an apparent horizon. Here we assume that such an apparent horizon is absent or $\widetilde{k}-\text{tr}_{ \widehat{\Sigma}}P>0$. Under such an assumption, a solution to Jang's equation exists and the foregoing construction makes sense. Under such circumstances, one can relate the generalized total mean curvature of $\Sigma$ in $\mathbb{R}^{1,3}$ to that of $ \widehat{\Sigma}$ in $\widehat{\Omega}$ together with the momentum information. We first state the theorem regarding the solvability of Jang's equation. The proof is presented in section 4 of \cite{yau}. As most estimates are derived in Schoen-Yau's original articles \cite{schoen1979proof,schoen1981proof} for the asymptotically flat case, it suffices to control the boundary gradient of the solution.  
\begin{theorem}\cite{yau}
 The normal derivative of a solution of the Dirichlet problem of Jang's equation is bounded if $\widetilde{k}>|\tr_{\widehat{\Sigma}}P|$, where $\widetilde{k}$ is the mean curvature of $\widehat{\Sigma}$ in $\widehat{\Omega}$.  
\end{theorem}
This theorem guarantees the existence of the desired hypersurface $\widehat{\Omega}$ in $\Omega\times \mathbb{R}$. Therefore we have the surface $\widehat{\Sigma}$ which is the boundary of the hypersurface $\widehat{\Omega}$ and isometric to $\widehat{i}_{0}(\widehat{\Sigma})$ i.e., we can define the isometric embedding $\widehat{i}_{0}:\widehat{\Sigma}\hookrightarrow \mathbb{R}^{3},~\widehat{\Sigma}\mapsto \widehat{i}_{0}(\widehat{\Sigma})$. Recall, here $\widehat{i}_{0}(\widehat{\Sigma})$ is the projection of the $i_{0}(\Sigma)\subset \mathbb{R}^{1,3}$ onto $T_{0}$ complement i.e., $\mathbb{R}^{3}$. Before moving onto proving the vital proposition \ref{mean2}, we define the $\tau$ that is admissible.   

\begin{definition}[Admissible Time function $\tau$]
\label{admi}
Given a spacelike embedding $i: \Sigma\hookrightarrow M$, a smooth function $\tau$ on $ \Sigma$ is said to be admissible if\\
(a) $K+\frac{\det\nabla^{2}\tau}{1+|\nabla\tau|^{2}}>0$,~(b) $\Sigma$ bounds an embedded space-like three manifold $\Omega$ in $(M,\widehat{g})$ such that Jang's equation \ref{eq:jang} with Dirichlet boundary data $\tau$ is solvable on $\Omega$, (c) The generalized mean curvature $h(\Sigma,i,\tau,e^{'}_{3})>0$ for a space-like unit normal $e^{'}_{3}$ is determined by Jang's equation. 

\end{definition}
Let $\tau$ be a smooth function on $\Sigma=\partial \Omega$. We consider a solution $f$ of Jang's equation in $\Omega\times \mathbb{R}$ that satisfies the Dirichlet boundary condition $f=\tau$ on $\Sigma$. Denote the graph of $\tau$ over $\Sigma$ by $\widehat{\Sigma}$ and that of $f$ over $\Omega$ by $\widehat{\Omega}$ so that $\widehat{\Sigma}=\partial\widehat{\Omega}$. We choose orthonormal frames $\{e_{1},e_{2}\}$ and $\{\widetilde{e}_{1},\widetilde{e}_{2}\}$ for $T\Sigma$ and $T\widehat{\Sigma}$, respectively. Let $e_{3}$ be the outward normal of $\Sigma$ that is tangent to $\Omega$. We choose $\widetilde{e}_{3},\widetilde{e}_{4}$ for the normal bundle of $\widehat{\Sigma}$ in $\Omega\times \mathbb{R}$ such that $\widetilde{e}_{3}$ is tangent to the graph of $\widehat{\Omega}$ and $\widetilde{e}_{4}$ is a downward unit normal vector of $\widehat{\Omega}$ in $\Omega\times \mathbb{R}$. Also, let $v$ be a downward unit vector to $\Omega$ in $\mathbb{R}$ direction. $\{e_{1},e_{2},e_{3},v\}$ forms an orthonormal basis for the tangent space of $\Omega\times \mathbb{R}$ and so does $\{\widetilde{e}_{\alpha}\}_{\alpha=1}^{4}$. All these frames are extended along $\mathbb{R}$ direction by parallel translation.     


\begin{proposition}
\label{mean2}
Let $\tau$ be a smooth function on $\Sigma=\partial\Omega$. We consider a
solution $f$ of Jang’s equation in $\Omega\times \mathbb{R}$ that satisfies the Dirichlet boundary condition $f=\tau$ on $\Sigma$. Denote the graph of $\tau$ over $\Sigma$ by $ \widehat{\Sigma}$. Then there exists a unit space-like normal field $e^{'}_{3}$ to $\Sigma$ in $M$ such that the following holds 
\begin{eqnarray}
\int_{ \widehat{\Sigma}}\left(\widetilde{k}-\langle\widetilde{\nabla}_{\widetilde{e}_{4}}\widetilde{e}_{4},\widetilde{e}_{3}\rangle+P(\widetilde{e}_{4},\widetilde{e}_{3})\right)\mu_{\widehat{\sigma}}=\int_{ \Sigma}\left(-\sqrt{1+|\nabla\tau|^{2}}\langle \mathbf{H},e^{'}_{3}\rangle-\alpha_{e^{'}_{3}}(\nabla\tau) \right)\mu_{\sigma} 
\end{eqnarray}
where $\widetilde{k}$ is the mean curvature of $ \widehat{\Sigma}$ in $\widehat{\Omega}$ and $\widetilde{e}_{4}$ and $\widetilde{e}_{3}$ are the basis of the normal bundle of $ \widehat{\Sigma}$ in $(M,\widehat{g})$. Let $e_{3}$ be the outward unit normal of $\Sigma$ that is tangent to $\Omega$ and $e_{4}$ is the future directed time-like normal of $\Omega$ in $(M,\widehat{g})$. $e^{'}_{3}$ is given by 
\begin{eqnarray}
 e^{'}_{3}=\cosh\varphi e_{3}+\sinh\varphi e_{4},~\sinh\varphi=-\frac{f_{3}}{\sqrt{1+|\nabla\tau|^{2}}},   
\end{eqnarray}
$f_{3}$ is the derivative of $f$ in $e_{3}$ direction along $\Sigma$. 
\end{proposition}
\begin{proof}

The proof is essentially a sequence of computations. We first prove the following identity. 
\begin{eqnarray}
\widetilde{k}-\langle\widetilde{\nabla}_{\widetilde{e}_{4}}\widetilde{e}_{4},\widetilde{e}_{4},\widetilde{e}_{3} \rangle+P(\widetilde{e}_{3},\widetilde{e}_{4})\\
=\langle\widetilde{\nabla}_{e_{I}}\widetilde{e}_{3} ,e_{I}\rangle+\frac{\langle e_{3},\widetilde{e}_{4}\rangle}{\langle e_{3},\widetilde{e}_{3}\rangle}\langle\widetilde{\nabla}_{e_{I}}\widetilde{e}_{4},e_{I}\rangle-\frac{\langle e_{3},\widetilde{e}_{4}\rangle}{\langle e_{3},\widetilde{e}_{3}\rangle}P(e_{I},e_{I})+\frac{1}{\langle e_{3},\widetilde{e}_{3}}P(e_{3},\widetilde{e}_{4}-\langle \widetilde{e}_{4},e_{3}\rangle e_{3}). 
\end{eqnarray}
This is done as follows
\begin{eqnarray}
 \langle\widetilde{\nabla}_{e_{I}}\widetilde{e}_{3},e_{I}\rangle=\sum_{j=1}^{3}\langle \widetilde{\nabla}_{e_{j}}\widetilde{e}_{3},e_{j}\rangle-\langle \widetilde{\nabla}_{e_{3}}\widetilde{e}_{3},e_{3}\rangle=\sum_{\alpha=1}^{4}\langle\widetilde{\nabla}_{\widetilde{e}_{\alpha}}\widetilde{e}_{3},\widetilde{e}_{\alpha}\rangle-\langle\widetilde{\nabla}_{e_{3}}\widetilde{e}_{3},e_{3}\rangle   
\end{eqnarray}
as $\{e_{1},e_{2},e_{3},\nu\}$ and $\{\widetilde{e}_{\alpha}\}_{\alpha=1}^{4}$ are both orthonormal frames for the tangent space of $\widehat{\Omega}\times \mathbb{R}$ and $\widetilde{\nabla}_{v}\widetilde{e}_{3}=0$. Now note 
\begin{eqnarray}
\sum_{I=1}^{2}\langle \widetilde{\nabla}_{e_{I}}\widetilde{e}_{3},e_{I}\rangle=\widetilde{k}+\langle\widetilde{\nabla}\widetilde{e}_{4}\widetilde{e}_{3},\widetilde{e}_{4}\rangle-\langle\widetilde{\nabla}_{\widetilde{e}_{3}}\widetilde{e}_{3},\widetilde{e}_{3}\rangle.
\end{eqnarray}
On the other hand we have 
\begin{eqnarray}
\sum_{I=1}^{2}\langle\widetilde{\nabla}_{e_{I}}\widetilde{e}_{4},e_{I}\rangle=\sum_{j=1}^{3}\langle\widetilde{\nabla}_{e_{j}}\widetilde{e}_{4},e_{j}\rangle-\langle\widetilde{\nabla}_{e_{3}}\widetilde{e}_{4},e_{3}\rangle=\sum_{j=1}^{3}\langle\widetilde{\nabla}_{\widetilde{e}_{j}}\widetilde{e}_{4},\widetilde{e}_{j}\rangle-\langle\widetilde{\nabla}_{e_{3}}\widetilde{e}_{4},e_{3}\rangle.    
\end{eqnarray}
Now, the application of Jang's equation yields 
\begin{eqnarray}
 \sum_{I=1}^{2}\langle\widetilde{\nabla}_{e_{I}}\widetilde{e}_{4},e_{I}\rangle=\sum_{j=1}^{3}P(\widetilde{e}_{j},\widetilde{e}_{j})-\langle\widetilde{\nabla}_{e_{3}}\widetilde{e}_{4},e_{3}\rangle.      
\end{eqnarray}
Moreover using 
\begin{eqnarray}
\sum_{j=1}^{3}P(\widetilde{e}_{j},\widetilde{e}_{j})=\sum_{\alpha=1}^{4}P(\widetilde{e}_{\alpha},\widetilde{e}_{\alpha})-P(\widetilde{e}_{4},\widetilde{e}_{4})=\sum_{j=1}^{3}P(e_{j},e_{j})-P(\widetilde{e}_{4},\widetilde{e}_{4})    
\end{eqnarray}
and 
\begin{eqnarray}
 \langle e_{3},\widetilde{e}_{4}\rangle P(\widetilde{e}_{4},\widetilde{e}_{4})=P(e_{3}-\langle e_{3},\widetilde{e}_{3}\rangle \widetilde{e}_{3},\widetilde{e}_{4})   
\end{eqnarray}
we have 
\begin{eqnarray}
    \sum_{j=1}^{3}P(\widetilde{e}_{j},\widetilde{e}_{j})=\sum_{j=1}^{3}P(e_{j},e_{j})+\frac{\langle e_{3},\widetilde{e}_{3}\rangle}{\langle e_{3},\widetilde{e}_{4}\rangle}P(\widetilde{e}_{3},\widetilde{e}_{4})-\frac{1}{\langle e_{3},\widetilde{e}_{4}\rangle}P(e_{3},\widetilde{e}_{4})
\end{eqnarray}
The final expression reads 
\begin{eqnarray}
\sum_{I=1}^{2}\langle \widetilde{\nabla}_{e_{I}}\widetilde{e}_{4},e_{I}\rangle=\sum_{I=1}^{2}P(e_{I},e_{I})+\frac{\langle e_{3},\widetilde{e}_{3}\rangle}{\langle \widetilde{e}_{3},e_{4}\rangle}P(\widetilde{e}_{3},\widetilde{e}_{4})-\frac{1}{\langle e_{3},\widetilde{e}_{4}\rangle}P(e_{3},\widetilde{e}_{4}-\langle \widetilde{e}_{4},e_{3}\rangle e_{3})-\langle \widetilde{\nabla}_{e_{3}}\widetilde{e}_{4},e_{3}\rangle    
\end{eqnarray}
Now let $k$ be the mean curvature of $ \widehat{\Sigma}$ in $\widehat{\Omega}$ with respect to $e_{3}$ i.e., 
\begin{eqnarray}
    k=\sum_{I=1}^{2}\langle \widetilde{\nabla}_{e_{I}}e_{3},e_{I}\rangle.
\end{eqnarray}
As $\langle e_{3},\widetilde{e}_{I}\rangle=0$, $e_{3}=\langle e_{3},\widetilde{e}_{3}\rangle e_{3}+\langle e_{3},\widetilde{e}_{4}\rangle \widetilde{e}_{4}$. With this, the expression of $k$ becomes 
\begin{eqnarray}
 k=\langle e_{3},\widetilde{e}_{3}\rangle\langle\widetilde{\nabla}_{e_{I}}\widetilde{e}_{3},e_{I}\rangle  +\langle e_{3},\widetilde{e}_{4}\rangle \langle \widetilde{\nabla}_{e_{I}}\widetilde{e}_{4},e_{I}\rangle-\langle e_{3},\widetilde{e}_{3}\rangle e_{I}(\langle e_{I},\widetilde{e}_{3}\rangle) -\langle e_{3},\widetilde{e}_{4}\rangle e_{I}(\langle e_{I},\widetilde{e}_{4}\rangle)
\end{eqnarray}
which yields the following relation between $k$ and $\widetilde{k}$
\begin{eqnarray}
\widetilde{k}-\langle \widetilde{\nabla}_{\widetilde{e}_{4}}\widetilde{e}_{4},\widetilde{e}_{3}\rangle+P(\widetilde{e}_{3},\widetilde{e}_{4})\\\nonumber 
=\frac{1}{\langle e_{3},\widetilde{e}_{3}\rangle}k-\frac{\langle e_{3},\widetilde{e}_{4}\rangle}{\langle e_{3},\widetilde{e}_{3}\rangle}P(e_{I},e_{I})+\frac{1}{\langle e_{3},\widetilde{e}_{3}\rangle}P(e_{3},\widetilde{e}_{4}-\langle \widetilde{e}_{4},e_{3}\rangle e_{3})\\\nonumber 
+e_{I}(\langle e_{I},\widetilde{e}_{3}\rangle)+\frac{\langle e_{3},\widetilde{e}_{4}\rangle}{\langle e_{3},\widetilde{e}_{3}\rangle}e_{I}(\langle e_{I},\widetilde{e}_{4}\rangle).
\end{eqnarray}
Observe 
\begin{eqnarray}
P(e_{3},\widetilde{e}_{4}-\langle \widetilde{e}_{4},e_{3}\rangle e_{3})=\frac{1}{\sqrt{1+|Df|^{2}}}P(e_{3},\nabla\tau)    
\end{eqnarray}
and 
\begin{eqnarray}
e_{I}(\langle e_{I},\widetilde{e}_{3}\rangle)+\frac{\langle e_{3},\widetilde{e}_{4}\rangle}{\langle e_{3},\widetilde{e}_{3}\rangle}e_{I}(\langle e_{I},\widetilde{e}_{4}\rangle)=-\frac{1}{\sqrt{1+|Df|^{2}}}\nabla\tau\cdot \nabla(\frac{f_{3}}{\sqrt{1+|\nabla\tau|^{2}}})    
\end{eqnarray}
which is nothing but 
\begin{eqnarray}
 \frac{\nabla\tau\cdot \nabla\varphi}{\sqrt{1+|\nabla\tau|^{2}}}.   
\end{eqnarray}
Therefore 
\begin{eqnarray}
\frac{1}{\langle e_{3},\widetilde{e}_{3}\rangle}k-\frac{\langle e_{3},\widetilde{e}_{4}\rangle}{\langle e_{3},\widetilde{e}_{3}\rangle}P(e_{I},e_{I})+\frac{1}{\langle e_{3},\widetilde{e}_{3}\rangle}P(e_{3},\widetilde{e}_{4}-\langle \widetilde{e}_{4},e_{3}\rangle e_{3})\\\nonumber 
+e_{I}(\langle e_{I},\widetilde{e}_{3}\rangle)+\frac{\langle e_{3},\widetilde{e}_{4}\rangle}{\langle e_{3},\widetilde{e}_{3}\rangle}e_{I}(\langle e_{I},\widetilde{e}_{4}\rangle)\\
=\frac{1}{\sqrt{1+|\nabla\tau|^{2}}}\left(\sqrt{1+|Df|^{2}}k-f_{3}P(e_{I},e_{I})+P(e_{3},\nabla\tau)+\nabla\tau\cdot\nabla\varphi\right)
\end{eqnarray}
and consequently 
\begin{eqnarray}
\widetilde{k}-\langle \widetilde{\nabla}_{\widetilde{e}_{4}}\widetilde{e}_{4},\widetilde{e}_{3}\rangle+P(\widetilde{e}_{3},\widetilde{e}_{4})\\\nonumber 
=\frac{1}{\sqrt{1+|\nabla\tau|^{2}}}\left(\sqrt{1+|Df|^{2}}\langle\nabla_{e_{I}}e_{3},e_{I}\rangle -f_{3}\langle\nabla_{e_{I}}e_{4},e_{I}\rangle-\alpha_{e_{3}}(\nabla\tau)+\nabla\tau\cdot\nabla\varphi\right)
\end{eqnarray}
On the other hand 
\begin{eqnarray}
 e^{'}_{3}=\cosh\varphi e_{3}+\sinh\varphi e_{4},~e^{'}_{4}=\sinh\varphi e_{3}+\cosh\varphi e_{4}   
\end{eqnarray}
yielding 
\begin{eqnarray}
  \langle \nabla_{e_{I}}e^{'}_{3},e_{I}\rangle-\frac{1}{\sqrt{1+|\nabla\tau|^{2}}}\alpha_{e^{'}_{3}}(\nabla\tau)\\\nonumber 
  =\cosh\varphi\langle\nabla_{e_{I}}e_{3},e_{I}\rangle+\sinh\varphi\langle\nabla_{e_{I}}e_{4},e_{I}\rangle-\frac{1}{\sqrt{1+|\nabla\tau|^{2}}}(\alpha_{3}(\nabla\tau)-\nabla\tau\cdot\nabla\varphi)
\end{eqnarray}
This completes the proof.
\end{proof}

\section{Proof of the theorem \ref{1}}
\label{proof1}
\noindent Let us recall the statement of the theorem \ref{1}\\

\noindent \textbf{Theorem 1.1:}
\textit{Let $(\Sigma, \sigma)$ be a spacelike topological $2-$sphere that bounds a spacelike topological ball $\Omega$ in the physical spacetime $(M,\widehat{g})$ verifying dominant energy condition \ref{eq:dominant}. Moreover, assume $\Sigma$ is not an apparent horizon in $(M,\widehat{g})$, and its mean curvature vector $\mathbf{H}$ is spacelike. Let $i_{0}:\Sigma\to \mathbb{R}^{1,3}$ be an isometric embedding into the Minkowski space and let $\tau$ denote the restriction of the time function $t$ on $i_{0}(\Sigma)$. Let $\Sigma,\Omega, \widehat{\Sigma}$, and $\widehat{\Omega}$ be as in definition \ref{surfaces}. Assume $\tau$ is admissible (see definition \ref{admi}) and $\mathcal{S}^{MIT}$ is non-empty then for a $\Psi\in \mathcal{S}^{MIT}$,
(a)\begin{eqnarray}
M^{\Psi}=\frac{1}{8\pi}\int_{ \widehat{\Sigma}}\left(k_{0}-(\widetilde{k}-\langle\widetilde{\nabla}_{\widetilde{e}_{4}}\widetilde{e}_{4},\widetilde{e}_{3}\rangle+P(\widetilde{e}_{4},\widetilde{e}_{3}))\right)|\Psi|^{2}\mu_{\widehat{\sigma}}   
\end{eqnarray}
is non-negative. Moreover $M^{\Psi}$ vanishes if and only if the spacetime $(M,\widehat{g})$ is Minkowski along $ \widehat{\Sigma}$,
(b) If $\mathcal{S}^{APS}$ is non-empty, then $M^{WY}$ is non-negative and $M^{WY}$ vanishes if and only if the spacetime $(M,\widehat{g})$ is Minkowski along $ \widehat{\Sigma}$}.

 The propositions \ref{mean1} and \ref{mean2} can now be applied to reduce the Wang-Yau quasi-local energy to the desired form on which the spinor argument can be applied. Let us consider the topological $2-$sphere $(\Sigma, \sigma)$ that bounds a spacelike ball $\Sigma$ in the spacetime $(M,\widehat{g})$. Let $\tau$ be an admissible time function (according to the definition \ref{admi}), then the metric on $\widehat{\Sigma}$ is $ \widehat{\sigma}:= \sigma+d\tau\otimes d\tau$. The mean curvature of $ \widehat{\Sigma}$ in $\widehat{\Omega}$ is $\widetilde{k}$. Now consider its isometric embedding into $\mathbb{R}^{3}$ $\widehat{i}_{0}: \widehat{\Sigma}\hookrightarrow \mathbb{R}^{3}$. Let the mean curvature of $\widehat{i}_{0}( \widehat{\Sigma})$ in $\mathbb{R}^{3}$ be $k_{0}$. Let $(\widetilde{e}_{3},\widetilde{e}_{4})$ be an orthonormal frame of the normal bundle of $ \widehat{\Sigma}$ as in proposition \ref{mean2}. According to the propositions \ref{mean1} and \ref{mean2}, the reduced Wang-Yau energy functional $M^{WY}_{\tau}$ reads (see definition \ref{eq:wy2})
\begin{eqnarray}
 8\pi M^{wy}_{\tau}:=\int_{ \widehat{\Sigma}}k_{0}\mu_{\widehat{\sigma}}-\int_{ \widehat{\Sigma}}\left(\widetilde{k}-\langle\widetilde{\nabla}_{\widetilde{e}_{4}}\widetilde{e}_{4},\widetilde{e}_{3}\rangle+P(\widetilde{e}_{4},\widetilde{e}_{3})\right)\mu_{\widehat{\sigma}},   
\end{eqnarray}
Let us define the vector field $X$ whose dual $1-$form reads 
\begin{eqnarray}
 X^{\flat}:=\langle\widetilde{\nabla}_{\widetilde{e}_{4}}\widetilde{e}_{4},\cdot\rangle-P(\widetilde{e}_{4},\cdot).  
\end{eqnarray}
 Also note that $\widetilde{e}_{3}$ is the unit spacelike normal to $ \widehat{\Sigma}$ in $\widehat{\Omega}$. Let $\nu:=\widetilde{e}_{3}$. Therefore, the Wang-Yau energy takes the form 
 \begin{eqnarray}
8\pi M^{wy}_{\tau}=\int_{ \widehat{\Sigma}}k_{0}\mu_{\widehat{\sigma}}-\int_{ \widehat{\Sigma}}(\widetilde{k}-\langle X,\nu\rangle)\mu_{\widehat{\sigma}}.     
 \end{eqnarray}
First, we prove the part $(a)$. Let $\Psi\in C^{\infty}(S^{\widehat{\Omega}})$ such that it solves the Dirac equation with MIT Bag boundary condition 
\begin{eqnarray}
\label{eq:boundary10}
\widehat{D}\Psi=0~on~\widehat{\Omega}\\
\Pi_{+}\Psi=\Pi_{+}\xi~on~ \widehat{\Sigma},
\end{eqnarray}
for $\xi\in C^{\infty}(S^{ \widehat{\Sigma}})$. We want to apply theorem \ref{dirac1}. First, we check the curvature condition of $\widehat{\Omega}$ in $(M,\widehat{g})$.  Now recall the geometry of $\widehat{\Omega}$ in the product $\Omega\times \mathbb{R}$. On $\widehat{\Omega}$ the induced metric is $\widetilde{g}_{ij}:=g_{ij}+\partial_{i}f\partial_{j}f$ and let us denote its second fundamental form by $h_{ij}$ (note that $h_{\mu\nu}=\frac{1}{2}(\widetilde{\nabla}_{\mu}\widetilde{n}_{\mu}+\widetilde{\nabla}_{\nu}\widetilde{n}_{\mu})$,~$\mu,\nu=1,2,3,4$, where $\widetilde{n}=\widetilde{e}_{4}$ is the unit downward time-like normal of $\widehat{\Omega}$ in $\Omega\times \mathbb{R}$). Similarly, let $P_{ij}$ be the second fundamental form of $\widehat{\Omega}$ in $\Omega\times \mathbb{R}$. On $\widehat{\Omega}$, one has the following inequality verified by the scalar curvature $R^{\widehat{\Omega}}$ if the Hamiltonian and the momentum constraints are satisfied on $\widehat{\Omega}$
\begin{eqnarray}
 2(\mu-\sqrt{\widetilde{g}(J,J)})\leq R^{\widehat{\Omega}}-\sum_{i,j}(h_{ij}-P_{ij})^{2}-2\sum_{i}(h_{i4}-P_{i4})^{2}+2\sum_{i}\widetilde{\nabla}_{i}(h_{i4}-P_{i4})   
\end{eqnarray}
But since the spacetime verifies the dominant energy condition, we have $\mu\geq \sqrt{\widetilde{g}(J,J)}$, one has 
\begin{eqnarray}
 R^{\widehat{\Omega}}\geq \sum_{i,j}(h_{ij}-P_{ij})^{2}+2\sum_{i}(h_{i4}-P_{i4})^{2}-2\sum_{i}\widetilde{\nabla}_{i}(h_{i4}-P_{i4})\\\nonumber 
 \geq 2\sum_{i}(h_{i4}-P_{i4})^{2}-2\sum_{i}\widetilde{\nabla}_{i}(h_{i4}-P_{i4}). 
\end{eqnarray}
But note that $h_{i4}-P_{i4}$ is nothing but the $\widetilde{e}_{i}$th components of the $1-$form $-P(\widetilde{e}_{4},\cdot)+\langle \widetilde{\nabla}_{\widetilde{e}_{4}}\widetilde{e}_{4},\cdot\rangle$. Therefore $\sum_{i}(h_{i4}-P_{i4})^{2}=|X|^{2}$ and $\sum_{i}\widetilde{\nabla}_{i}(h_{i4}-P_{i4})=\text{div}X$. Therefore, under the assumption of the dominant energy condition, we have 
\begin{eqnarray}
R^{\widehat{\Omega}}\geq 2|X|^{2}-2\text{div}X.    
\end{eqnarray}
Therefore, the theorem \ref{dirac1} applies and we immediately have a unique solution $\Psi$ for the boundary value problem \ref{eq:boundary10} given a boundary spinor $\xi\in C^{\infty}(S^{ \widehat{\Sigma}})$. Now according to theorem \ref{MIT1}, $\Psi$ verifies the inequality 
\begin{eqnarray}
\label{eq:live}
\int_{ \widehat{\Sigma}}\langle D\Psi,\Psi\rangle \mu_{\widetilde{ \widehat{\Sigma}}}\geq \frac{1}{2}\int_{ \widehat{\Sigma}}(\widetilde{k}-\langle X,\nu\rangle)|\Psi|^{2}\mu_{\widehat{\sigma}}.    
\end{eqnarray}
Now let's consider the isometric image $\widehat{i}_{0}( \widehat{\Sigma})$ of $ \widehat{\Sigma}$ in $\mathbb{R}^{3}$. Let's assume that $\widehat{i}_{0}( \widehat{\Sigma})$ bounds a topological ball $\overline{\widehat{\Omega}}$ in $\mathbb{R}^{3}$. Since $\mathbb{R}^{3}$ has parallel spinors, the following boundary value problem for the Dirac equation on $\overline{\widehat{\Omega}}$ 
\begin{eqnarray}
\label{eq:R3}
 \mathbf{D}\Phi=0,~on~\overline{\widehat{\Omega}}\subset \mathbb{R}^{3}\\
 \Pi_{+}\Phi=\Pi_{+}\zeta~on~\widehat{i}_{0}( \widehat{\Sigma})
\end{eqnarray}
has a unique solution given $\zeta\in C^{\infty}(S^{\overline{\widehat{\Omega}}})$. Here $\mathbf{D}$ is the Dirac operator on $\mathbb{R}^{3}$ induced by the canonical flat connection. Moreover $\Phi$ verifies the equality due to theorem \ref{MIT1} and 
\begin{eqnarray}
D\Phi=\frac{k_{0}}{2}\Phi~on~\widehat{i}_{0}( \widehat{\Sigma})=\partial \overline{\widehat{\Omega}}.    
\end{eqnarray}
Here $D$ is the intrinsic Dirac operator on the boundary $\widehat{i}_{0}( \widehat{\Sigma})$ which is identical to the Dirac operator $D$ on $ \widehat{\Sigma}$ since $ \widehat{\Sigma}$ and $\widehat{i}_{0}( \widehat{\Sigma})$ are isometric and therefore has identical intrinsic Riemannian structure. 

Now we pull back the restriction of $\Phi$ onto $\widehat{i}_{0}( \widehat{\Sigma})$ onto $ \widehat{\Sigma}$ by the isometry map and set $\xi\in C^{\infty}(S^{ \widehat{\Sigma}})$ equation to  $\widehat{i}^{*}_{0}\Phi$ i.e., 
\begin{eqnarray}
 \xi= \widehat{i}^{*}_{0}\Phi~on ~ \widehat{\Sigma}.  
\end{eqnarray}
Note that $|\widehat{i}^{*}_{0}\Phi|=|\Phi|$. Now on $ \widehat{\Sigma}$, 
\begin{eqnarray}
\label{eq:quality}
 D\xi=\frac{k_{0}}{2}\xi   
\end{eqnarray}
since on $\widehat{i}_{0}( \widehat{\Sigma})$ 
\begin{eqnarray}
D\Phi=\frac{k_{0}}{2}\Phi.    
\end{eqnarray}
Unfortunately, we only have $\Pi_{+}\Psi=\Pi_{+}\xi$ on $ \widehat{\Sigma}$. But the chirality decomposition of the spinor $\Psi$ gives 
\begin{eqnarray}
 \Psi=\Psi_{+}+\Psi_{-}   
\end{eqnarray}
where we denote $\Pi_{\pm}\Psi$ by $\Psi_{\pm}$ following the notation in corollary \ref{MITC} and the MIT Bag boundary condition implies 
\begin{eqnarray}
\Psi_{+}=\xi_{+}    
\end{eqnarray}
and therefore 
\begin{eqnarray}
D\Psi_{+}=D\xi_{+}    
\end{eqnarray}
since $D$ is $ \widehat{\Sigma}$ parallel. 
The inequality \ref{eq:live} reads 
\begin{eqnarray}
\int_{ \widehat{\Sigma}}\left(\langle D\Psi_{+},\Psi_{-}\rangle +\langle D\Psi_{-},\Psi_{+}\rangle\right)\mu_{\widehat{\sigma}}\geq \frac{1}{2}\int_{ \widehat{\Sigma}}(\widetilde{k}-\langle X,\nu\rangle)|\Psi|^{2}\mu_{\widehat{\sigma}}   
\end{eqnarray}
since $DS^{ \widehat{\Sigma}}_{+}=S^{ \widehat{\Sigma}}_{-}$ and the sections of these two subbundles are pointwise orthogonal with respect to the Hermitian inner product being used here.
Now using the boundary condition $\Psi_{+}=\xi_{+}$ and self-adjointness of $D$ on closed $ \widehat{\Sigma}$, we have 
\begin{eqnarray}
\label{eq:qual2}
\int_{ \widehat{\Sigma}}2\langle D\xi_{+},\Psi_{-}\rangle \mu_{\widehat{\sigma}}\geq \frac{1}{2}\int_{ \widehat{\Sigma}}(\widetilde{k}-\langle X,\nu\rangle)|\Psi|^{2}\mu_{\widehat{\sigma}}.    
\end{eqnarray}
But now from \ref{eq:quality}
\begin{eqnarray}
D\xi=\frac{k_{0}}{2}\xi~or~D\xi_{+}=\frac{k_{0}}{2}\xi_{-}    
\end{eqnarray}
since $D\Pi_{+}\xi=\Pi_{-}D\xi$. Therefore \ref{eq:qual2} becomes 
\begin{eqnarray}
 \int_{ \widehat{\Sigma}}k_{0}\langle \xi_{-},\Psi_{-}\rangle \mu_{\widehat{\sigma}}\geq \frac{1}{2}\int_{ \widehat{\Sigma}}(\widetilde{k}-\langle X,\nu\rangle)|\Psi|^{2}\mu_{\widehat{\sigma}}.   
\end{eqnarray}
Now since the mean curvature $\widehat{k}_{0}$ of $\widehat{i}_{0}( \widehat{\Sigma})$ is positive, we have by Cauchy-Scwartz
\begin{eqnarray}
  \frac{1}{2}\int_{ \widehat{\Sigma}}k_{0}(|\xi_{-}|^{2}+|\Psi_{-}|^{2})\mu_{\widehat{\sigma}}\geq \int_{ \widehat{\Sigma}}k_{0}\langle \xi_{-},\Psi_{-}\rangle \mu_{\widehat{\sigma}}\geq \frac{1}{2}\int_{ \widehat{\Sigma}}(\widetilde{k}-\langle X,\nu\rangle)|\Psi|^{2}\mu_{\widehat{\sigma}}.  
\end{eqnarray}
But since $D\xi=\frac{k_{0}}{2}\xi$ on $ \widehat{\Sigma}$, we can apply corollary \ref{MITC} to conclude 
\begin{eqnarray}
\label{eq:114}
\frac{1}{2}\int_{ \widehat{\Sigma}}k_{0}(|\xi_{-}|^{2}+|\Psi_{-}|^{2})\mu_{\widehat{\sigma}}=\frac{1}{2}\int_{ \widehat{\Sigma}}k_{0}(|\xi_{+}|^{2}+|\Psi_{-}|^{2})\mu_{\widehat{\sigma}}    
\end{eqnarray}
but $\xi_{+}=\Psi_{+}$ on $ \widehat{\Sigma}$ and therefore 
\begin{eqnarray}
\label{eq:119}
\frac{1}{2}\int_{ \widehat{\Sigma}}k_{0}(|\xi_{+}|^{2}+|\Psi_{-}|^{2})\mu_{\widehat{\sigma}}=\frac{1}{2}\int_{ \widehat{\Sigma}}k_{0}(|\Psi_{+}|^{2}+|\Psi_{-}|^{2})\mu_{\widehat{\sigma}} =\frac{1}{2}\int_{ \widehat{\Sigma}}k_{0}|\Psi|^{2}\mu_{\widehat{\sigma}}.    
\end{eqnarray}
Therefore \ref{eq:114} and \ref{eq:119} yields 
\begin{eqnarray}
 \frac{1}{2}\int_{ \widehat{\Sigma}}\left(k_{0}-(\widetilde{k}-\langle X,\nu\rangle\right)|\Psi|^{2}\mu_{\widehat{\sigma}}\geq 0.   
\end{eqnarray}
This proves that the spinorial quasi-local energy 
\begin{eqnarray}
\label{eq:R4}
 M^{\Psi}=\frac{1}{8\pi}\int_{ \widehat{\Sigma}}\left(k_{0}-(\widetilde{k}-\langle X,\nu\rangle\right)|\Psi|^{2}\mu_{\widehat{\sigma}}\geq 0.   
\end{eqnarray}
This completes the first part of the proof. Now, the reduced Wang-Yau energy is achieved from the spinor quasi-local energy $M^{\Psi}$ as 
\begin{eqnarray}
M^{WY}_{\tau}=\inf_{\Psi\in \mathcal{S}^{MIT},|\Psi|=1}M^{\Psi} \geq 0.   
\end{eqnarray}
This holds for any admissible $\tau$. Therefore, $M^{WY}\geq 0$. 

For the part $(b)$, we need to use theorem \ref{dirac} and corollary \ref{spininequality} i.e., work with APS boundary condition. We use the same symbols to denote the spinors but care must be taken while interpreting: they are completely different from the previous case. Let $\Psi\in C^{\infty}(S^{ \widehat{\Sigma}})$. Now As usual consider the boundary value problem 
\begin{eqnarray}
\widehat{D}\Psi=0~on~\widehat{\Omega}\\
P_{\geq 0}\Psi=P_{\geq 0} \xi~on \widehat{ \widehat{\Sigma}}    
\end{eqnarray}
which has a unique solution given $\xi\in C^{\infty}(S^{ \widehat{\Sigma}})$ due to theorem \ref{dirac}. Moreover, under the condition on the point-wise norm $|P_{\geq 0}\Psi|\leq |\Psi|$ on $ \widehat{\Sigma}$, and the fact that $\widetilde{k}>\langle X,\nu\rangle$, $\Psi$ verifies the inequality \ref{spininequality} 
\begin{eqnarray}
\label{eq:129}
\int_{ \widehat{\Sigma}}\left(\langle DP_{\geq 0}\Psi,P_{\geq 0}\Psi\rangle-\frac{1}{2}(\widetilde{k}-\langle X,\nu\rangle)|P_{\geq 0}\Psi|^{2}\right) \mu_{\widehat{\sigma}}\geq 0.    
\end{eqnarray} 
Now let us solve the Dirac equation on $\overline{\widehat{\Omega}}$ in $\mathbb{R}^{3}$ bounded by the isometric image $\widehat{i}_{0}( \widehat{\Sigma}_{0})$ of $ \widehat{\Sigma}_{0}$. Let $\mathbf{D}$ be the Dirac operator on $\mathbb{R}^{3}$ induced by the canonical flat connection. According to the theorem \ref{dirac}, for $\Phi\in C^{\infty}(S^{\overline{\widehat{\Omega}}})$, the boundary value problem 
\begin{eqnarray}
\mathbf{D}\Phi=0 ~on~\overline{\widehat{\Omega}},\\
P_{\geq 0}\Phi=P_{\geq 0}\zeta~on~\widehat{i}_{0}( \widehat{\Sigma})
\end{eqnarray}
for a $\zeta\in C^{\infty}(S^{\widehat{i}_{0}( \widehat{\Sigma})})$. Now according to corollary \ref{spininequality}, on $\overline{\widehat{\Omega}}\subset \mathbb{R}^{3}$, $\Phi$ verifies the equality case in the inequality \ref{eq:129} and consequently 
\begin{eqnarray}
 \Phi=P_{\geq 0}\Phi   
\end{eqnarray}
on $\widehat{i}_{0}( \widehat{\Sigma})$ and $\Phi$ is aparallel spinor on $\overline{\widehat{\Omega}}$. Since $\Phi$ is parallel on $\overline{\widehat{\Omega}}$, it verifies 
\begin{eqnarray}
 \nabla|\Phi|^{2}=\langle\widehat{\nabla}\Phi,\Phi\rangle+\langle\Phi,\widehat{\nabla}\Phi\rangle=0   
\end{eqnarray}
and therefor without loss of generality, we can set 
\begin{eqnarray}
 |\Phi|=1.   
\end{eqnarray}
But since by corollary \ref{spininequality}, $\Phi=P_{\geq 0}\Phi$ on $\widehat{i}_{0}( \widehat{\Sigma})$, we have 
\begin{eqnarray}
 |P_{\geq 0}\Phi|=1~on~ \widehat{i}_{0}( \widehat{\Sigma}).  
\end{eqnarray}
Moreover, also due to the corollary \ref{spininequality}, on $\widehat{i}_{0}( \widehat{\Sigma})$, 
\begin{eqnarray}
 DP_{\geq 0}\Phi=\frac{k_{0}}{2}P_{\geq 0}\Phi,~on~\widehat{i}_{0}( \widehat{\Sigma}).   
\end{eqnarray}
Now pull back the restriction of $\Phi$ onto $\widehat{i}_{0}( \widehat{\Sigma})$ i.e., $\Phi=P_{\geq 0}\Phi$ onto $ \widehat{\Sigma}$ by the isometry map and set $\xi\in C^{\infty}(S^{ \widehat{\Sigma}})$ equal to $\widehat{i}^{*}_{0}\Phi$ i.e., 
\begin{eqnarray}
 \xi=\widehat{i}^{*}_{0}\Phi~on~ \widehat{\Sigma}.   
\end{eqnarray}
Note that $|P_{\geq 0}\xi|=|P_{\geq 0}\Phi|=1$. Now since on $\widehat{i}_{0}( \widehat{\Sigma})$, 
\begin{eqnarray}
    DP_{\geq 0}\Phi=\frac{k_{0}}{2}P_{\geq 0}\Phi,
\end{eqnarray}
we have on $ \widehat{\Sigma}$
\begin{eqnarray}
\label{eq:xi}
 DP_{\geq 0}\xi=\frac{k_{0}}{2}P_{\geq 0}\xi.   
\end{eqnarray}
Now on $ \widehat{\Sigma}$, the inequality 
\begin{eqnarray}
\int_{ \widehat{\Sigma}}\left(\langle DP_{\geq 0}\Psi,P_{\geq 0}\Psi\rangle-\frac{1}{2}(\widetilde{k}-\langle X,\nu\rangle)|P_{\geq 0}\Psi|^{2}\right) \mu_{\widehat{\sigma}}\geq 0    
\end{eqnarray}
with the boundary condition $P_{\geq 0}\Psi=P_{\geq 0\xi}$ on $ \widehat{\Sigma}$ yields
\begin{eqnarray}
 \int_{ \widehat{\Sigma}}\left(\langle DP_{\geq 0}\xi,P_{\geq 0}\xi\rangle-\frac{1}{2}(\widetilde{k}-\langle X,\nu\rangle)|P_{\geq 0}\xi|^{2}\right) \mu_{\widehat{\sigma}}\geq 0.   
\end{eqnarray}
Now use \ref{eq:xi} to conclude
\begin{eqnarray}
\int_{ \widehat{\Sigma}}\left(\frac{k_{0}}{2}|P_{\geq 0}\xi|^{2}-\frac{1}{2}(\widetilde{k}-\langle X,\nu\rangle)|P_{\geq 0}\xi|^{2}\right) \mu_{\widehat{\sigma}}\geq 0    
\end{eqnarray}
and therefore by $|P_{\geq 0}\xi|=1$ 
\begin{eqnarray}
 \int_{ \widehat{\Sigma}}\left(\frac{k_{0}}{2}-\frac{1}{2}(\widetilde{k}-\langle X,\nu\rangle)\right) \mu_{\widehat{\sigma}}\geq 0.   
\end{eqnarray}
Therefore the reduced Wang-Yau energy
\begin{eqnarray}
M^{WY}_{\tau}:=\frac{1}{8\pi}\int_{ \widehat{\Sigma}}(k_{0}-(\widetilde{k}-\langle X,\nu\rangle))\mu_{\widehat{\sigma}}\geq 0    
\end{eqnarray}
which holds for every admissible $\tau$. Therefore minimizing within the admissible set of $\tau$s, we obtain  
\begin{eqnarray}
 M^{WY}\geq 0.   
\end{eqnarray}
This completes the proof of the theorem \ref{1}.

\begin{corollary}[Rigidity]
 Under the assumption of the theorem \ref{embedding} and if the set of admissible $\tau$ is non-empty, then the Wang-Yau quasi-local energy bounded by the spacelike topological $2-$sphere $ \Sigma$ is non-negative. Moreover, if $\Sigma$ is a spacelike topological $2-$sphere embedded in $\mathbb{R}^{1,3}$, then Wang-Yau quasi-local energy $M^{wy}$ is zero. Also, if $M^{wy}$ is zero then the physical spacetime $N$ is Minkowski along $ \widehat{\Sigma}$.
\end{corollary}
\begin{proof}
The proof of the non-negativity is trivial if the admissible set of $\tau$ is non-empty and essentially a consequence of the proof of the main theorem stated earlier in this section. If the topological $2-$sphere $\Sigma$ is embedded in $\mathbb{R}^{1,3}$, then one can simply take the embedding $i_{0}$ to be $i$ and therefore the vanishing of $M^{wy}$ follows trivially. Now if $M^{wy}$ vanishes, then we have $\widehat{\Omega}$ to be Ricci flat and therefore flat and more specifically $R^{\widehat{\Omega}}=0$. Now if the equality holds in \ref{spininequality}, then according to the theorem \ref{spininequality}, we have $X=0$. Therefore $h_{i4}=P_{i4}$ which yields together with $R^{\widehat{\Omega}}=0$ 
\begin{eqnarray}
 0\leq 2(\mu-|J|)\leq -\sum_{i,j}(h_{ij}-P_{ij})^{2}\leq 0
\end{eqnarray}
and therefore $\mu=|J|$ and $h_{ij}=P_{ij}$. Therefore following a similar argument as of \cite{schoen1981proof}, it follows that $(M,\widehat{g})$ is Minkowski along $\Omega$.
\end{proof}

\begin{corollary}
Liu-Yau mass is non-negative.    
\end{corollary}
\begin{proof}
Set $\tau=\text{constant}$. Then the Wang-Yau quasi-local energy reduces to Liu-Yau energy. The embedding $i_{0}$ is in $\mathbb{R}^{3}$. Then one needs a stronger condition of strict positivity of the Gauss curvature of $ \widehat{\Sigma}$ i.e., $K>0$. Now the result follows from the main theorem. 
\end{proof}

\section{Proof of Theorem \ref{2}}
\noindent Proof of the positivity of the Brown-York mass is actually simpler. The strategy is similar to that of the thereom \ref{1}. Therefore, we do not repeat all the details. In this case, there is no difference between the surfaces $\Sigma$ and $\widehat{\Sigma}$ since we consider embedding into $\mathbb{R}^{3}$. First, we focus on part (a). 

Let $ \widehat{\Sigma}$ be a topological $2-$ sphere in the spacetime $(M,\widehat{g})$ bounding a topological ball $\widehat{\Omega}$. The given condition is the non-negativity of the scalar curvature of $\widehat{\Omega}$ i.e.,
\begin{eqnarray}
 R^{\widehat{\Omega}}\geq 0.  
\end{eqnarray}
Now consider the isometric embedding of $ \widehat{\Sigma}$ into the Euclidean space $\mathbb{R}^{3}$ $i: \widehat{\Sigma}\hookrightarrow \mathbb{R}^{3}$ and denote the image by $i( \widehat{\Sigma})$. By the hypothesis on the strict positivity of the Gaussian curvature of $ \widehat{\Sigma}$, such an embedding exists by \cite{pogorelov1952regularity}. Let the mean curvature of $ \widehat{\Sigma}$ in $\widehat{\Omega}$ be $k$ while that of its isometric image in $\mathbb{R}^{3}$ is $k_{0}$. Proceeding exactly similar manner, we can prove a theorem analogous to \ref{MIT1} by taking $X=0$. More precisely, we can solve for the Dirac equation on $\widehat{\Omega}$ with MIT Bag boundary condition. Let $\Psi\in C^{\infty}(S^{\widehat{\Omega}})$ such that it solves the Dirac equation with MIT Bag boundary condition 
\begin{eqnarray}
\label{eq:boundary10}
\widehat{D}\Psi=0~on~\widehat{\Omega}\\
\Pi_{+}\Psi=\Pi_{+}\xi~on~ \widehat{\Sigma},
\end{eqnarray}
for $\xi\in C^{\infty}(S^{ \widehat{\Sigma}})$. Proceeding in an exactly similar manner, we prove under the non-negative of the curvature, this problem has a unique non-trivial solution given $\xi$. Consequently, the following inequality follows 
\begin{eqnarray}
 \int_{ \widehat{\Sigma}}\langle D\Psi,\Psi\rangle \mu_{\widehat{\sigma}}\geq \frac{1}{2}\int_{ \widehat{\Sigma}}k|\Psi|^{2}\mu_{\widehat{\sigma}},  
\end{eqnarray}
where the equality holds when $ \widehat{\Sigma}$ is embedded in $\mathbb{R}^{3}$ and $\Psi$ is the restriction to $ \widehat{\Sigma}$ of a non-trivial parallel spinor in $\widehat{\Omega}\subset \mathbb{R}^{3}$. Now proceeding in an exact similar way as in \ref{eq:R3}-\ref{eq:R4} we obtain 
\begin{eqnarray}
\widetilde{M}^{\Psi}:=\frac{1}{8\pi}\int_{ \widehat{\Sigma}}(k_{0}-k)|\Psi|^{2}\mu_{\widehat{\sigma}}\geq 0    
\end{eqnarray}
and consequently 
\begin{eqnarray}
M^{BY}=\inf_{\Psi\in \mathcal{S}^{MIT},|\Psi|=1}\widetilde{M}^{\Psi}\geq 0.    
\end{eqnarray}
Part (b) follows in an exact similar manner as in part (b) of the theorem \ref{1}.

\section{Conclusions}
We provided a complete quasi-local proof of the non-negativity of the Wang-Yau quasi-local energy and its rigidity property by using spinors that solve the Dirac equation on the bulk with MIT bag (local) and Atyiah-Patodi-Singer (APS) boundary conditions. The analogous result for the Brown-York mass is proved too. In the asymptotically flat setting, the Wang-Yau quasi-local mass approaches ADM mass at spacelike infinity and its positivity proof using spinor (after Schoen-Yau \cite{schoen1979proof,schoen1981proof} proof was announced) in that circumstance was provided by Witten \cite{witten1981new} which was subsequently made rigorous by Parker-Taubes \cite{parker1982witten}. We would like to understand if the spinors can be used to define quasi-local mass contained within a compact spacelike hypersurface with a disconnected boundary, such as a topologically spherical annulus as well as in the presence of gauge fields. Another interesting problem would be to obtain a spinorial quasi-local proof of the positivity and rigidity of the new quasi-local mass defined by Alaee-Khuri-Yau\cite{alaee2023quasi} where the boundary surface $ \widehat{\Sigma}$ is allowed to have non-trivial topology. These problems are under investigation.

\section{Acknowledgement} This work was supported by the Center of Mathematical Sciences and Applications (CMSA), Department of Mathematics at Harvard University. We thank Christian B\"ar for helpful discussions regarding the boundary value problem associated with the Dirac equation. 

\section{Funding and/or Conflicts of interests/Competing interests}
The authors have no conflicts to disclose.

\end{document}